\documentclass[a4paper,11pt]{article}

\usepackage{amsmath,amsfonts,amssymb,amsthm,amscd}
\usepackage{graphicx,mathtools}
\usepackage{perpage}
\usepackage{url}
\usepackage{color}
\usepackage[T1]{fontenc}
\usepackage{hyperref}
\usepackage[a4paper,scale={0.72,0.74},marginratio={1:1},footskip=7mm,headsep=10mm]{geometry}
\usepackage[latin9]{inputenc}
\usepackage{mathrsfs}
\usepackage{tikz}
\usepackage{caption}
\usepackage{subcaption}
\usepackage{float}
\usepackage{enumerate}
\usepackage[english]{babel}
\usepackage{hyperref}
\usepackage{setspace}
\usepackage{rotating}
\makeatletter
\def\@captype{figure}
\makeatother

\numberwithin{equation}{section}

\newtheorem{theorem}{Theorem}[section]
\newtheorem{lemma}[theorem]{Lemma}

\newtheorem{remark}[theorem]{Remark}
\newtheorem{definition}[theorem]{Definition}

\newtheorem{conjecture}[theorem]{Conjecture}
\newtheorem{approximation}[theorem]{Approximation}

\newcommand{\R}{\mathbb{R}}

\newcommand{\N}{\mathbb{N}}

\newcommand{\eee}{\mathrm{e}}
\newcommand{\ddd}{\mathrm{d}}
\newcommand{\cT}{\mathcal{T}}
\newcommand{\im}{\mathrm{i}}

\begin{document}

%%%%%%%%%% TITLE PAGE %%%%%%%%%%%%%%

\title{Synchronization of phase oscillators\\ 
on the hierarchical lattice}

\author{
\renewcommand{\thefootnote}{\arabic{footnote}}
D.\ Garlaschelli 
\footnotemark[1]
\\
\renewcommand{\thefootnote}{\arabic{footnote}}
F.\ den Hollander
\footnotemark[2]
\\
\renewcommand{\thefootnote}{\arabic{footnote}}
J.\ Meylahn
\footnotemark[2]
\\
\renewcommand{\thefootnote}{\arabic{footnote}}
B.\ Zeegers
\footnotemark[2]
}

\footnotetext[1]{
Lorentz Institute for Theoretical Physics, Leiden University, P.O.\  Box 9504, 
2300 RA Leiden, The Netherlands.
}
\footnotetext[2]{
Mathematical Institute, Leiden University, P.O.\ Box 9512,
2300 RA Leiden, The Netherlands.
}

\date{\today}

\maketitle

%%%%%%%%%%%%%% ABSTRACT %%%%%%%%%%%%%%%%%%%%%%%% 

\begin{abstract}
Synchronization of neurons forming a network with a hierarchical structure is essential for 
the brain to be able to function optimally. In this paper we study synchronization of phase 
oscillators on the most basic example of such a network, namely, the hierarchical lattice. 
Each site of the lattice carries an oscillator that is subject to noise. Pairs of oscillators interact 
with each other at a strength that depends on their hierarchical distance, modulated by a 
sequence of interaction parameters. We look at block averages of the oscillators on successive 
hierarchical scales, which we think of as block communities. In the limit as the number of 
oscillators per community tends to infinity, referred to as the hierarchical mean-field limit, 
we find a separation of time scales, i.e., each block community behaves like a single oscillator 
evolving on its own time scale. We argue that the evolution of the block communities is given 
by a renormalized mean-field noisy Kuramoto equation, with a synchronization level that 
depends on the hierarchical scale of the block community. We find three universality 
classes for the synchronization levels on successive hierarchical scales, characterized 
in terms of the sequence of interaction parameters. 

What makes our model specifically challenging is the non-linearity of the interaction between
the oscillators. The main results of our paper therefore come in three parts: (I) a \emph{conjecture} 
about the nature of the renormalisation transformation connecting successive hierarchical 
scales; (II) a \emph{truncation approximation} that leads to a simplified renormalization transformation; 
(III) a \emph{rigorous analysis} of the simplified renormalization transformation. We provide 
compelling arguments in support of (I) and (II), but a full verification remains an open problem. 



\medskip\noindent
{\it Acknowledgment.} 
DG is  supported by EU-project 317532-MULTIPLEX. FdH and JM are supported by NWO 
Gravitation Grant 024.002.003--NETWORKS. The authors are grateful to G.\ Giacomin for 
critical remarks. 

\end{abstract}

\newpage

%%%%%%%%%%%%% SECTION 1 %%%%%%%%%%%%%%%%%%%%%%

\section{Introduction}

The concept of \emph{spontaneous synchronization} is ubiquitous in nature. Single 
oscillators (like flashing fireflies, chirping crickets or spiking brain cells) may rotate 
incoherently, at their own natural frequency, when they are isolated from the population, 
but within the population they adapt their rhythm to that of the other oscillators, acting 
as a system of coupled oscillators. There is no central driving mechanism, yet the 
population reaches a globally synchronized state via mutual local interactions.

The omnipresence of spontaneous synchronization triggered scientists to search for a 
mathematical approach in order to understand the underlying principles. The first steps 
were taken by Winfree~\cite{W67}, \cite{W80}, who recognized that spontaneous 
synchronization should be understood as a threshold phenomenon: if the coupling 
between the oscillators is sufficiently strong, then a macroscopic part of the population 
freezes into synchrony. Although the model proposed by Winfree was too difficult 
to solve analytically, it inspired Kuramoto ~\cite{K75}, \cite{K84} to suggest a 
more mathematically tractable model that captures the same phenomenon. The 
Kuramoto model has since been used successfully to study synchronization in a 
variety of different contexts. By now there is an extended literature, covering aspects 
like phase transition, stability, and effect of disorder (for a review, see Ac\'ebron 
\emph{et al.}~\cite{A05}).

Mathematically, the Kuramoto model still poses many challenges. As long as the 
interaction is \emph{mean-field} (meaning that every oscillator interacts equally 
strongly with every other oscillator), a fairly complete theory has been developed.
However, as soon as the interaction has a \emph{non-trivial geometry}, computations 
become cumbersome. There is a large literature for the Kuramoto model on complex 
networks, where the population is viewed as a random graph whose vertices carry 
the oscillators and whose edges represent the interaction. Numerical and heuristic 
results have been obtained for networks with a small-world, scale-free and/or community 
structure, showing a range of interesting phenomena (for a review, see Arenas
\emph{et al.}~\cite{A08}). Rigorous results are rare. In the present paper we focus 
on one particular network with a community structure, namely, the \emph{hierarchical lattice}. 

The remainder of this paper is organised as follows. Sections~\ref{sec:MFK}--\ref{sec:scaling} 
are devoted to the mean-field noisy Kuramoto model. In Section~\ref{sec:MFK} we recall 
definitions and basic properties. In Section~\ref{sec:MV} we recall the McKean-Vlasov equation, 
which describes the evolution of the probability density for the phase oscillators in the 
\emph{mean-field limit}. In Section~\ref{sec:scaling} we take a closer look at the scaling 
properties of the order parameters towards the mean-field limit. In Section~\ref{sec:HL} 
we define the hierarchical lattice and in Section~\ref{sec:HKM} introduce the noisy Kuramoto 
model on the hierarchical lattice, which involves a sequence of interaction strengths 
$(K_k)_{k\in\N}$ acting on successive hierarchical levels. Section~\ref{sec:results} contains 
our main results, presented in the form of a conjecture, a truncation approximation, and  
rigrorous theorems. These concern the \emph{hierarchical mean-field limit} and show that, 
for each $k\in\N$, the block communities at hierarchical level $k$ behave like the mean-field 
noisy Kuramoto model, with an interaction strength and a noise that depend on $k$ and are 
obtained via a \emph{renormalization transformation} connecting successive hierarchical 
levels. There are three universality classes for $(K_k)_{k\in\N}$, corresponding to sudden 
loss of synchronization at a finite hierarchical level, gradual loss of synchronization as the 
hierarchical level tends to infinity, and no loss of synchronization. The renormalization 
transformation allows us to describe these classes in some detail. In Section~\ref{sec:thmscalphase} 
we analyse the renormalization scheme, in Section~\ref{sec:thmclasses} we find criteria for 
the universality classes. Appendix~\ref{app:numerics} provides numerical examples and 
computations.

%%%%%%%%%%

\subsection{Mean-field Kuramoto model}
\label{sec:MFK}

We begin by reviewing the mean-field Kuramoto model. Consider a population of $N\in\N$ 
oscillators, and suppose that the $i^{\text{th}}$ oscillator has a natural frequency 
$\omega_{i}$, such that
\begin{equation}
\label{eq:frequencies}
\begin{aligned}
\blacktriangleright\quad\, 
&\omega_{i},\, i = 1,\ldots,N, \text{ are i.i.d.\ and are drawn from}\\[-.1cm] 
&\text{a common probability distribution } \mu \text{ on } \R.   
\end{aligned}  
\end{equation}
Let the phase of the $i^{\text{th}}$ oscillator at time $t$ be $\theta_{i}(t) \in \R$. 
If the oscillators were not interacting, then we would have the system of 
uncoupled differential equations 
\begin{equation}
\label{eq:Knoint}
\frac{\ddd \theta_{i}(t)}{\ddd t} = \omega_i, \quad  i = 1, \ldots, N.
\end{equation}
Kuramoto~\cite{K75}, \cite{K84} realized that the easiest way to allow for synchronization 
was to let every oscillator interact with every other oscillator according to the sine of their 
phase difference, i.e., to replace \eqref{eq:Knoint} by:
\begin{equation}
\label{eq:Knonoise}
\frac{\ddd \theta_{i}(t)}{\ddd t} = \omega_{i} 
+ \frac{K}{N} \sum_{j=1}^{N} \sin\big[\theta_{j}(t) - \theta_{i}(t)\big],
\quad i = 1, \ldots, N.
\end{equation} 
Here, $K \in (0,\infty)$ is the \emph{interaction strength}, and the factor $\frac{1}{N}$ is included 
to make sure that the total interaction per oscillator stays finite in the thermodynamic limit 
$N\to\infty$. The coupled evolution equations in \eqref{eq:Knonoise} are referred to as the 
\emph{mean-field Kuramoto model}. An illustration of the interaction in this model is given in 
Fig.~\ref{fig:mfkuramoto}. 

%%%%%%%%%%%%%%%%%%%%%%%%%%%%%%%%%%%%%%%%%%%
\begin{figure}[htbp]
\vspace{0.2cm}
\centering
\includegraphics[scale=0.4]{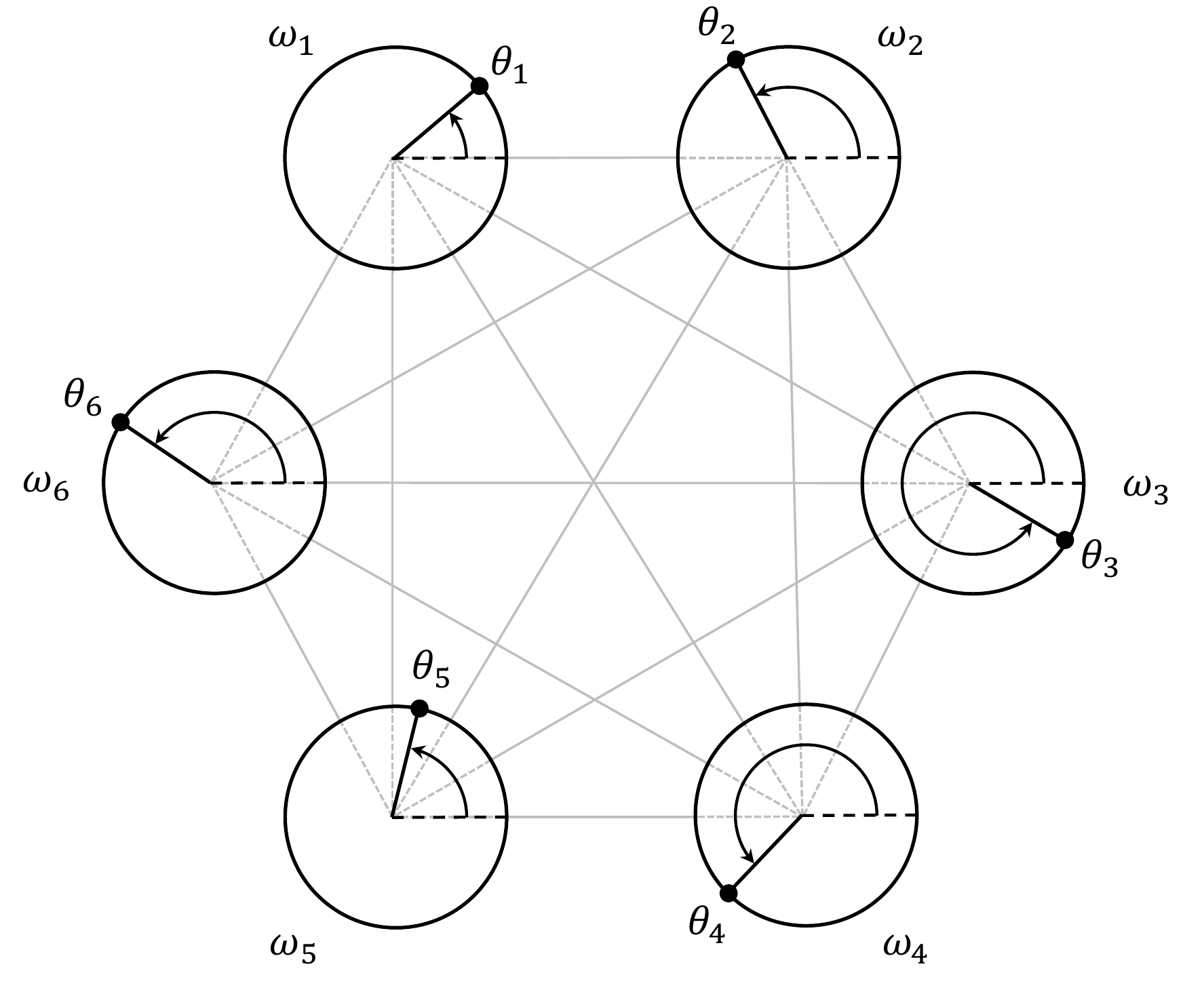}
\caption{Mean-field interaction of $N=6$ oscillators with natural frequencies $\omega_{i}$ 
and phases $\theta_{i}$, $i=1,\ldots,6$, evolving according to \eqref{eq:Knonoise}.}
\label{fig:mfkuramoto}
\end{figure}
%%%%%%%%%%%%%%%%%%%%%%%%%%%%%%%%%%%%%%%

\noindent
If noise is added, then \eqref{eq:Knonoise} turns into the \emph{mean-field noisy Kuramoto model}, 
given by
\begin{equation}
\label{eq:Knoise}
\ddd \theta_{i}(t) = \omega_{i}\,\ddd t 
+ \frac{K}{N} \sum_{j=1}^{N} \sin\big[\theta_{j}(t) - \theta_{i}(t)\big]\,\ddd t + D\,\ddd W_i(t), 
\quad i = 1, \ldots, N.
\end{equation} 
Here, $D \in (0,\infty)$ is the \emph{noise strength}, and $(W_i(t))_{t \geq 0}$, $i=1,\ldots,N$, are 
independent standard Brownian motions on $\R$. The coupled evolution equations in 
\eqref{eq:Knoise} are stochastic differential equations in the sense of It\^o (see e.g.\ Karatzas 
and Shreve~\cite{KS98}). As initial condition we take
\begin{equation}
\label{eq:initial}
\begin{aligned}
\blacktriangleright\quad\, 
&\theta_i(0),\, i = 1,\ldots,N, \text{ are i.i.d.\ and are drawn from}\\[-.1cm] 
&\text{a common probability distribution } \rho \text{ on } [0,2\pi).   
\end{aligned}  
\end{equation}

In order to exploit the mean-field nature of \eqref{eq:Knoise}, the complex-valued order 
parameter (with $\im$ the imaginary unit)
\begin{equation}
\label{eq:orderparfiniteN}
r_N(t)\,\eee^{\im\psi_N(t)} = \frac{1}{N} \sum_{j=1}^N \eee^{\im\theta_j(t)} 
\end{equation}
is introduced. In \eqref{eq:orderparfiniteN}, $r_N(t)$ is the \emph{synchronization level} at time
$t$ and takes values in $[0,1]$, while $\psi_N(t)$ is the \emph{average phase} at time $t$ and 
takes values in $[0,2\pi)$. (Note that $\psi_N(t)$ is properly defined only when $r_N(t)>0$.) The 
order parameter $(r,\psi)$ is illustrated in Fig.~\ref{fig:orderparameter} ($r=0$ corresponds to the 
oscillators being completely unsynchronized, $r=1$ to the oscillators being completely synchronized). 

%%%%%%%%%%%%%%%%%%%%%%%%%%%%%%%%%%%%%%%
\begin{figure}[htbp]
\centering
\begin{subfigure}[b]{0.4\linewidth}
\centering
\includegraphics[scale=0.25]{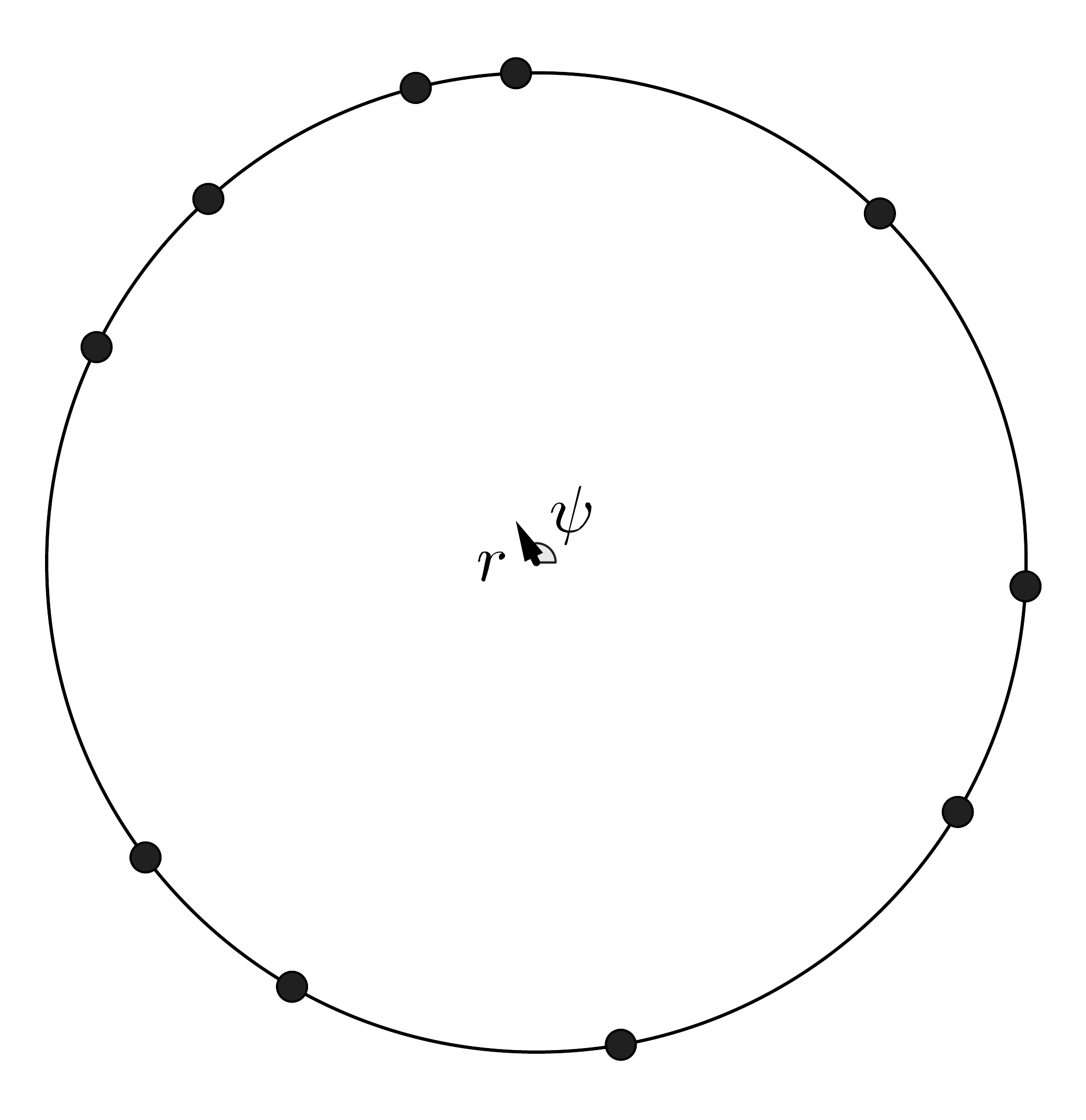}
\caption{$r = 0.095$.}
\end{subfigure}
\begin{subfigure}[b]{0.48\linewidth}
\centering
\includegraphics[scale=0.25]{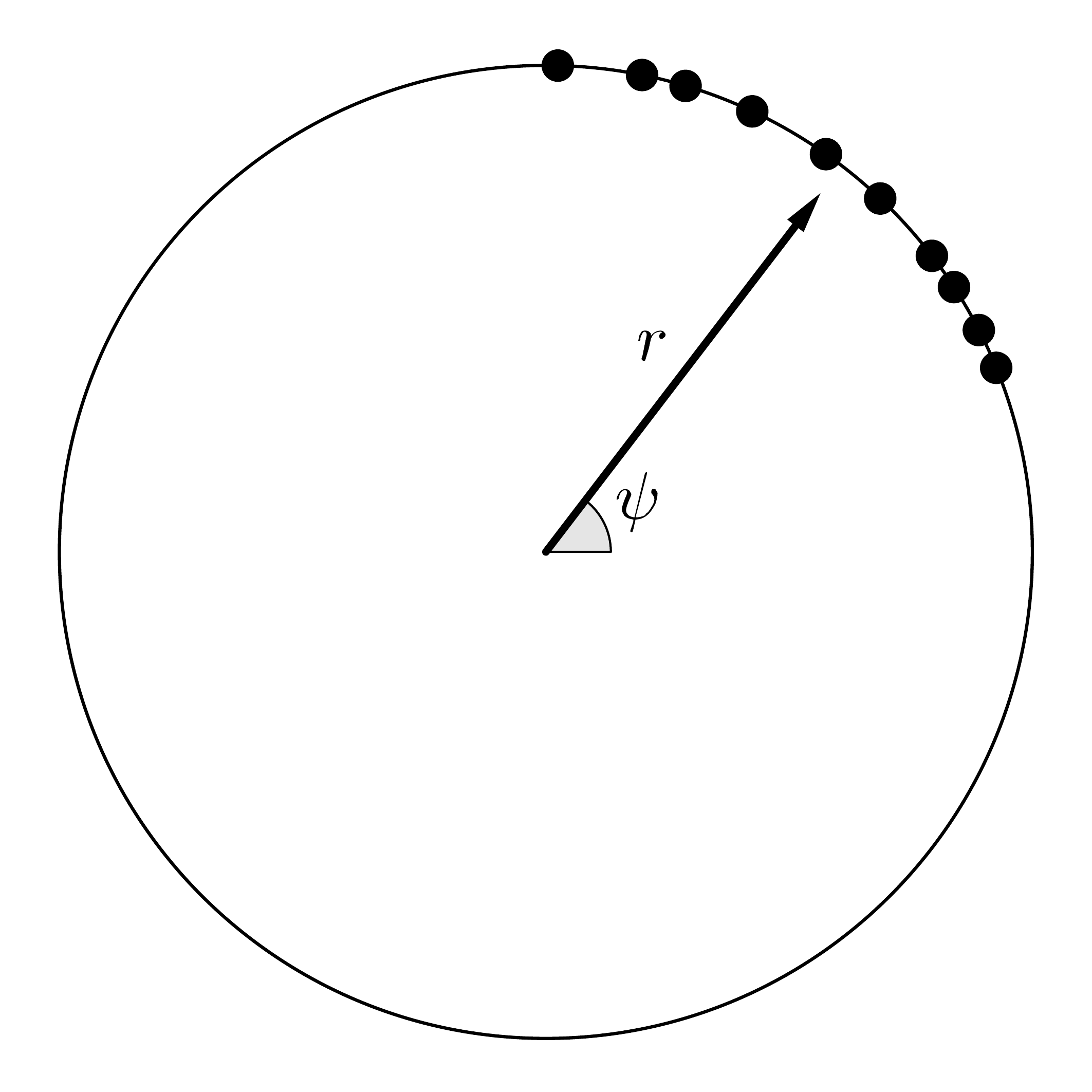}
\caption{$r = 0.929$.}
\end{subfigure}
\caption{Phase distribution of oscillators for two different values of $r$. The arrow represents
the complex number $r\eee^{i\psi}$.}
\label{fig:orderparameter}
\end{figure}
%%%%%%%%%%%%%%%%%%%%%%%%%%%%%%%%%%%%%%

\noindent
By rewriting \eqref{eq:Knoise} in terms of \eqref{eq:orderparfiniteN} as
\begin{equation}
\label{eq:Knoiseorderpar}
\ddd \theta_{i}(t) = \omega_{i}\,\ddd t 
+ K r_N(t) \sin\big[\psi_N(t)-\theta_{i}(t)\big]\,\ddd t + D\,\ddd W_i(t), 
\quad i = 1, \ldots, N,
\end{equation} 
we see that \emph{the oscillators are coupled via the order parameter}, i.e., the phases $\theta_i$ 
are pulled towards $\psi_N$ with a strength proportional to $r_N$. Note that $r_N(t)$ and $\psi_N(t)$
are random variables that depend on $\mu$, $D$ and $\rho$.  

In the \emph{mean-field limit} $N\to\infty$, the system in \eqref{eq:Knoiseorderpar} exhibits 
what is called ``propagation  of chaos'', i.e., the evolution of single oscillators becomes 
\emph{autonomous}. Indeed, let the order parameter associated with $\rho$ in \eqref{eq:initial} 
be the pair $(R,\Phi) \in [0,1] \times [0,2\pi)$ defined by
\begin{equation}
\label{RPhirho}
R\,\eee^{\im\Phi} = \int_0^{2\pi} \rho(\ddd\theta)\,\eee^{\im\theta}.
\end{equation}
Suppose that $R>0$, so that $\Phi$ is properly defined. Suppose further that
\begin{equation}
\label{eq:symuni}
\blacktriangleright\quad 
\text{the disorder distribution } \mu \text{ in \eqref{eq:frequencies} is symmetric}.
\end{equation}
Then, as we will see in Sections~\ref{sec:MV}--\ref{sec:scaling}, the limit as $N\to\infty$ 
of the evolution of a single oscillator, say $\theta_1$, is given by
\begin{equation}
\label{eq:propcha}
\ddd \theta_1(t) = \omega_1\,\ddd t 
+ K r(t) \sin\big[\Phi-\theta_1(t)\big]\,\ddd t + D\,\ddd W_1(t),
\end{equation} 
where $(W_1(t))_{t \geq 0}$ is a standard Brownian motion, and $r(t)$ is driven by a 
\emph{deterministic relaxation equation} such that 
\begin{equation}
\label{eq:rrelax}
r(0) = R, \qquad \lim_{t\to\infty} r(t) = r \text{ for some } r \in [0,1). 
\end{equation}
The parameter $r=r(\mu,D,K)$ will be identified in \eqref{mfKsyn} below (and the convergence 
holds at least when $R$ is close to $r$; see Remark~\ref{rem:stability} below). The evolution in 
\eqref{eq:propcha} is \emph{not closed} because of the presence of $r(t)$, but after a 
\emph{transient period} it converges to the \emph{autonomous} evolution equation
\begin{equation}
\label{eq:propchalt}
\ddd \theta_1(t) = \omega_1\,\ddd t 
+ K r \sin\big[\Phi-\theta_1(t)\big]\,\ddd t + D\,\ddd W_1(t).
\end{equation} 
Without loss of generality, we may \emph{calibrate} $\Phi=0$ by rotating the circle $[0,2\pi)$ 
over $-\Phi$. After that the parameters $R,\Phi$ associated the initial distribution $\rho$ are gone, 
and only $r$ remains as the relevant parameter. It is known (see e.g.\ \eqref{eq:Kcid} below) 
that there exists a \emph{critical threshold} $K_c = K(\mu,D)  \in (0,\infty)$ separating two regimes: 
\begin{itemize}
\item
For $K \in (0,K_c]$ the system relaxes to an \emph{unsynchronized state} ($r = 0$). 
\item
For $K \in (K_c,\infty)$ the system relaxes to a \emph{partially synchronized state} ($r \in (0,1)$), 
at least when $\rho$ in \eqref{eq:initial} is chosen such that $R$ is close to $r$ (see
Remark~\ref{rem:stability} below).  
\end{itemize}
See Strogatz~\cite{S00} and Lu\c{c}on~\cite{L12} for overviews.

%%%%%%%%%%%%%%%%%%%%%%%%%%%%%%%

\subsection{McKean-Vlasov equation}
\label{sec:MV}

For the system in \eqref{eq:Knoise}, Sakaguchi~\cite{S88} showed that in the limit as $N\to\infty$,
the probability density for the phase oscillators and their natural frequencies (with respect to 
$\lambda \times \mu$, with $\lambda$ the Lebesgue measure on $[0,2\pi]$ and $\mu$ the disorder 
measure on $\R$) evolves according to the \emph{McKean-Vlasov equation} 
\begin{equation}
\label{eq:McKean}
\frac{\partial}{\partial t}\,p(t;\theta,\omega)
= -\frac{\partial}{\partial \theta}\,\Big[p(t;\theta,\omega)\,\Big\{\omega 
+ Kr(t)\sin\big[\psi(t) - \theta\big]\Big\}\Big] 
+ \frac{D}{2}\frac{\partial^{2}}{\partial \theta^{2}}\,p(t;\theta,\omega),
\end{equation}  
where
\begin{equation}
\label{eq:orderpar}
r(t)\,\eee^{\im\psi(t)} = \int_{\R} \mu(\ddd\omega) 
\int_{0}^{2\pi} \ddd\theta\,\eee^{\im\theta}\,p(t;\theta,\omega),
\end{equation}
is the continuous counterpart of \eqref{eq:orderparfiniteN}. If $\rho$ has a density, say 
$\rho(\theta)$, then $p(0;\theta,\omega)=\rho(\theta)$ for all $\omega\in\R$.

By \eqref{eq:symuni}, we can again \emph{calibrate} the average phase to be zero, i.e., 
$\psi(t) = \psi(0) = \Phi = 0$, $t \geq 0$, in which case the stationary solutions of 
\eqref{eq:McKean} satisfy
\begin{equation}
\label{eq:McKVstat}
0 = -\frac{\partial}{\partial \theta}\,\big[p(\theta,\omega)\,(\omega - Kr\sin\theta)\big]
+ \frac{D}{2}\frac{\partial^{2}}{\partial \theta^{2}}\,p(\theta,\omega).
\end{equation}
The solutions of \eqref{eq:McKVstat} are of the form
\begin{equation}
\label{eq:stat}
p_\lambda(\theta,\omega) 
= \frac{A_\lambda(\theta,\omega)}{\int_{0}^{2\pi} \ddd\phi\,A_\lambda(\phi,\omega)},
\qquad \lambda = 2Kr/D,
\end{equation}
with
\begin{equation}
\label{eq:ABdef}
\begin{aligned}
A_\lambda(\theta,\omega) 
&= B_\lambda(\theta,\omega) \left(\eee^{4\pi \omega}\int_{0}^{2\pi} 
\frac{\ddd\phi}{B_\lambda(\phi,\omega)} 
+ (1-\eee^{4\pi\omega}) \int_{0}^{\theta}\frac{\ddd\phi}{B_\lambda(\phi,\omega)}\right),\\[0.2cm]
B_\lambda(\theta,\omega) 
&= \eee^{\lambda\cos\theta + 2\theta\omega}.
\end{aligned}
\end{equation}
After rewriting 
\begin{equation}
A_\lambda(\theta,\omega) 
= B_\lambda(\theta,\omega) \left(\int_{\theta-2\pi}^{0} 
\frac{\ddd\phi}{B_\lambda(-\phi,-\omega)} 
+  \int_{0}^{\theta}\frac{\ddd\phi}{B_\lambda(\phi,\omega)}\right)
\end{equation}
and noting that $B_\lambda(\phi,\omega)=B_\lambda(-\phi,-\omega)$, we easily check that 
\begin{equation}
p_\lambda(\theta,\omega) = p_\lambda(-\theta,-\omega),
\end{equation} 
a property we will need later. In particular, in view of \eqref{eq:symuni}, we have
\begin{equation}
\int_{\R} \mu(\ddd\omega) \int_{0}^{2\pi} \ddd\theta\,p_\lambda(\theta, \omega)\,\sin\theta = 0.
\end{equation}

Since $\psi(t) = \psi(0) = \Phi = 0$, we see from \eqref{eq:orderpar} that $p_\lambda(\theta,\omega)$ 
in \eqref{eq:stat} is a solution if and only if $r$ satisfies
\begin{equation}
\label{mfKsyn}
\int_{\R} \mu(\ddd\omega) \int_{0}^{2\pi} \ddd\theta\,p_\lambda(\theta,\omega)\,\cos\theta = r,
\qquad \lambda=2Kr/D.
\end{equation}
This gives us a \emph{self-consistency relation} for
\begin{equation}
r=r(D,K)
\end{equation} 
a situation that is typical for mean-field systems, which can in principle be solved (and possibly 
has more than one solution). The equation in \eqref{mfKsyn} always has a solution with $r=0$: 
the \emph{unsynchronized state} corresponding to $p_0(\theta,\omega) = \frac{1}{2\pi}$ for all 
$\theta,\omega$. A (not necessarily unique) solution with $r \in (0,1)$ exists when the coupling 
strength $K$ exceeds a critical threshold $K_{c}=K_c(\mu,D)$. When this occurs, we say that 
the oscillators are in a \emph{partially synchronized state}. As $K$ increases also $r$ increases 
(see Fig.~\ref{fig-r(K)}). Moreover, $r \uparrow 1$ as $K \to \infty$ and we say that the oscillators 
converge to a \emph{fully synchronized state}. When $K$ crosses $K_{c}$, the system exhibits 
a second-order phase transition, i.e., $K \mapsto r(K)$ is continuous at $K=K_c$.

%%%%%%%%%%%%%%%%%%%
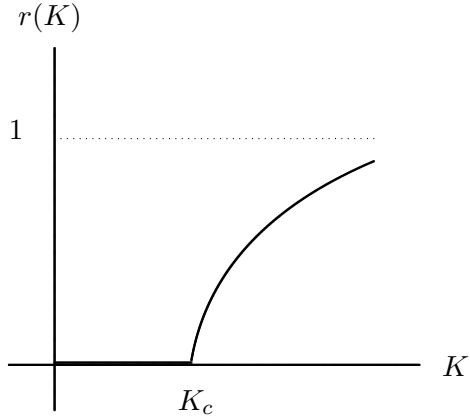
\begin{figure}[htbp]
\vspace{0.4cm}
\begin{center}
\setlength{\unitlength}{0.6cm}
\begin{picture}(8,6)(0,0)
%%%
{\thicklines
\qbezier(0,-1)(0,3)(0,7)
\qbezier(-1,0)(4,0)(8,0)
\qbezier(0,0.05)(1.5,0.05)(3,0.05)
\qbezier(3,0.1)(3.5,3)(7,4.5)
}
\qbezier[50](0,5)(3.5,5)(7,5)
\put(8.5,-.25){$K$}
\put(-.8,7.5){$r(K)$}
\put(2.7,-1){$K_c$}
\put(-1,5){$1$}
\end{picture}
\end{center}
\vspace{0.5cm}
\caption{\small Picture of $K \mapsto r(K)$ for fixed $\mu$ and $D$.}
\label{fig-r(K)}
\vspace{0.2cm}
\end{figure}
%%%%%%%%%%%%%%%%%%%%%%%%%%%%%%%

For the case where the frequency distribution $\mu$ is \emph{symmetric and unimodal}, an 
explicit expression is known for $K_{c}$:
\begin{equation}
\label{eq:Kcid}
\frac{1}{K_{c}} = \int_\R \mu(\ddd\omega)\,\frac{D}{D^2+4\omega^2}.
\end{equation}
Thus, when the spread of $\mu$ is large compared to $K$, the oscillators are not able to 
synchronize and they rotate near their own frequencies. As $K$ increases, this remains 
the case until $K$ reaches $K_{c}$. After that a small fraction of synchronized oscillators 
starts to emerge, which becomes of macroscopic size when $K$ moves beyond $K_c$.
For $\mu$ symmetric and unimodal it is \emph{conjectured} that for $K > K_c$ there is 
a \emph{unique} synchronized solution $p_\lambda(\cdot,\cdot)$ with $r \in (0,1)$ solving 
\eqref{mfKsyn} (Lu\c{c}on~\cite[Conjecture 3.12]{L12}). This conjecture has been proved 
when $\mu$ is narrow, i.e., the disorder is small (Lu\c{c}on~\cite[Proposition 3.13]{L12}).

\begin{remark}
\label{rem:stability}
{\rm
Stability of stationary solutions has been studied by Strogatz and Mirollo~\cite{SM91},
Strogatz, Mirollo and Matthews~\cite{SMM92}, Lu\c{c}on~\cite[Section 3.4]{L12}. For symmetric 
unimodal disorder, the unsynchronized state is linearly stable for $K<K_c$ and linearly unstable 
for $K>K_c$, while the synchronized state for $K>K_c$ is linearly stable at least for small disorder. 
Not much is known about stability for general disorder.} \hfill\qed
\end{remark}

There is no closed form expression for $K_{c}$ beyond symmetric unimodal disorder, except 
for special cases, e.g.\ symmetric binary disorder. We refer to Lu\c{c}on~\cite{L12} for an 
overview. A large deviation analysis of the empirical process of oscillators has been carried 
out in Dai Pra and den Hollander~\cite{DdH96}. 

%%%%%%%%%%%%%%%%%%%%%%%%%%%%%%%%%

\subsection{Diffusive scaling of the average phase}
\label{sec:scaling}

Bertini, Giacomin and Poquet~\cite{BGP14} showed that for the mean-field noisy Kuramoto 
model \emph{without disorder}, in the limit as $N\to\infty$ the synchronization level evolves 
on time scale $t$ and converges to a deterministic limit, while the average phase evolves on 
time scale $Nt$ and converges to a Brownian motion with a \emph{renormalized noise strength}. 
\footnote{The fact that the average phase evolves slowly was already noted by Ha and 
Slemrod~\cite{HS11} for the Kuramoto model with disorder and without noise, while an 
approximate solution was obtained by Sonnenschein and Schimansky-Geier~\cite{SS13} 
for the Kuramoto model without disorder and with noise.}

\begin{theorem}[Bertini, Giacomin and Poquet~\cite{BGP14}]
\label{thm:slowangle}
Suppose that $\mu = \delta_0$ and $r>0$. Then, in distribution, 
\begin{equation}
\label{eq:rpsitimescale}
\begin{aligned}
&\lim_{N\to\infty} \psi_N(Nt) = \psi_*(t),\\ 
&\lim_{N\to\infty} r_N(t) = r(t),
\end{aligned}
\end{equation}
with 
\begin{equation}
\begin{array}{lll}
&\ddd \psi_*(t) = D_*\,\ddd W_*(t), &\psi_*(0) = \Phi,\\
&\lim_{t\to\infty} r(t) = r, &r(0) = R,
\end{array}
\end{equation}
where $(W_*(t))_{t \geq 0}$ is a standard Brownian motion and
\begin{equation}
\label{D*def}
D_* = D_*(D,K, r) = \frac{1}{\sqrt{1-[I_0(2Kr/D)]^{-2}}}, \qquad r= r(D,K),
\end{equation}
with $I_0$ the modified Bessel function of order zero given by 
\begin{equation}
\label{eq:I0def}
I_0(\lambda) = \frac{1}{2\pi} \int_0^{2\pi}\ddd\theta\,\eee^{\lambda\cos\theta}, 
\qquad \lambda \in [0,\infty).
\end{equation} 
\end{theorem}

\noindent
The work in \cite{BGP14} also shows that
\begin{equation}
\label{eq:rNNt}
\lim_{N\to\infty} r_N(Nt)=r \qquad \forall\,t>0,
\end{equation} 
i.e., the synchronization level not only tends to $r$ over time, it also stays close to $r$ on a 
time scale of order $N$. Thus, the synchronization level is much less volatile than the average 
phase.

In Section~\ref{sec:prep} we explain the heuristics behind Theorem~\ref{thm:slowangle}. This
heuristics will play a key role in our analysis of the Kuramoto model on the hierarchical lattice 
in the hierarchical mean-field limit. In fact, Conjecture~\ref{thm:scalphasend} below will extend 
Theorem~\ref{thm:slowangle} to the hierarchical lattice. It is important to note that the diffusive 
scaling only occurs in the model \emph{without disorder}. Indeed, for the model with disorder it was 
shown in Lu\c{c}on and Poquet~\cite{LP17} that the fluctuations of the disorder prevail over the 
fluctuations of the noise, resulting in `travelling waves' for the empirical distribution of the 
oscillators. Therefore, also on the hierarchical lattice we only consider the model without disorder.  

%%%%%%%%%%%%%%%%%%%%%%%%%%%%%%%%%

\subsection{Hierarchical lattice}
\label{sec:HL}

The hierarchical lattice of order $N$ consist of countable many vertices that form communities 
of sizes $N$, $N^2$, etc. For example, the hierarchical lattice of order $N=3$ consists of vertices 
that are grouped into $1$-block communities of $3$ vertices, which in turn are grouped 
into $2$-block communities of $9$ vertices, and so on. Each vertex is assigned a label that 
defines its location at each block level (see Fig.~\ref{fig:hierarchical-lattice}).

%%%%%%%%%%%%%%%%%%%%%%%%%%%%%%%%%%%%%%%%%%
\begin{figure}[htbp] 
\centering
\includegraphics[width=0.85\linewidth]{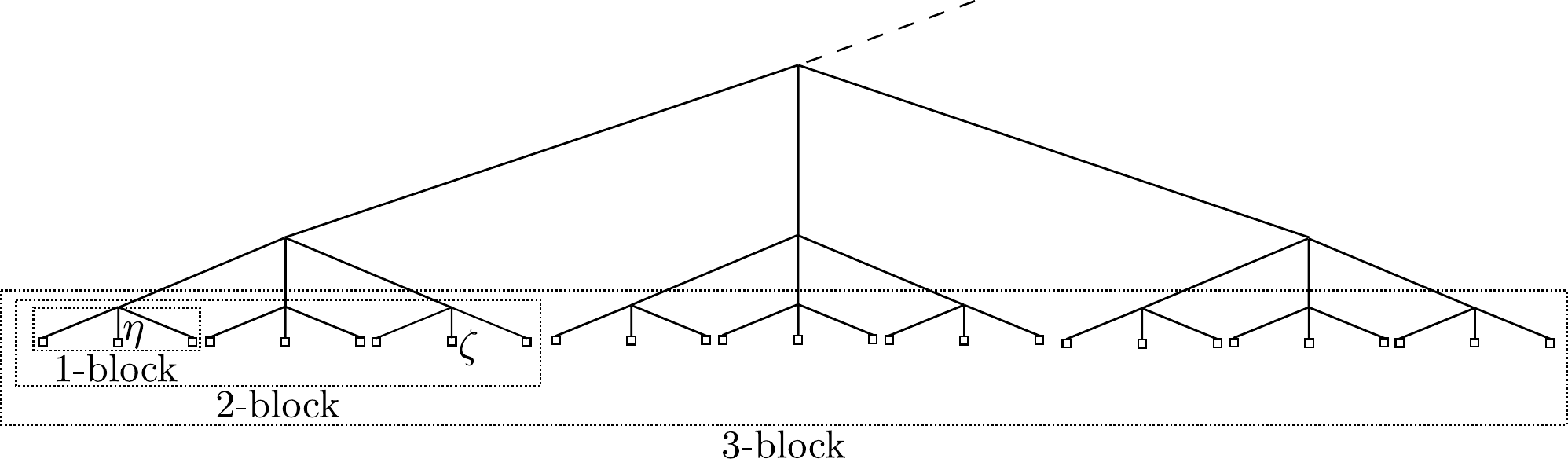} 
\caption{\small The hierarchical lattice of order $N=3$. The vertices live at the lowest level. 
The tree visualizes their distance: the distance between two vertices $\eta,\zeta$ is the height 
of their lowest common branching point in the tree: $d(\eta,\zeta)=2$ in the picture.} 
\label{fig:hierarchical-lattice} 
\end{figure}
%%%%%%%%%%%%%%%%%%%%%%%%%%%%%%%%%%%%%%%%%%%%%

Formally, the hierarchical group $\Omega_{N}$ of order $N\in\N\backslash\{1\}$ is the set
\begin{equation}
\Omega_{N} = \bigg\{\eta = (\eta^{\ell})_{\ell\in\N_{0}} \in \{0,1,\ldots,N-1\}^{\N_{0}}
\colon\, \sum_{\ell\in\N_{0}} \eta^{\ell} < \infty \bigg\}
\end{equation}
with addition modulo $N$, i.e., $(\eta + \zeta)^{\ell} = \eta^{\ell} + \zeta^{\ell}\,(\text{mod}\,N)$, 
$\ell\in\N_{0}$. The distance on $\Omega_{N}$ is defined as
\begin{equation}
d\colon\,\Omega_{N} \times \Omega_{N}\to\N_{0}, \quad (\eta,\zeta) \mapsto 
\min\big\{k\in\N_{0}\colon\,\eta^{\ell} = \zeta^{\ell}\,\,\forall\, \ell\geq k\big\},
\end{equation}
i.e., the distance between two vertices is the smallest index from which onwards the 
sequences of hierarchical labels of the two vertices agree. This distance is ultrametric:
\begin{equation}
d(\eta,\zeta) \leq \min\{d(\eta,\xi),d(\zeta,\xi)\} \qquad \forall\,\eta,\zeta,\xi \in \Omega_{N}.
\end{equation} 
For $\eta \in \Omega_{N}$ and $k\in\mathbb{N}_{0}$, the $k$-block around $\eta$ is 
defined as
\begin{equation}
B_{k}(\eta) = \{\zeta\in\Omega_{N}\colon\,d(\eta,\zeta) \leq k\}.
\end{equation}

%%%%%%%%%%%%%%%%%%%%%%%%%%%%%%%%%%%

\subsection{Hierarchical Kuramoto model}
\label{sec:HKM}

We are now ready to define the model that will be our object of study. Each vertex 
$\eta\in\Omega_{N}$ carries a phase oscillator, whose phase at time $t$ is denoted
by $\theta_{\eta}(t)$. Oscillators interact in pairs, but at a strength that depends on 
their hierarchical distance. To modulate this interaction, we introduce a sequence of 
interaction strengths
\begin{equation}
\label{eq:Kdef}
(K_k)_{k\in\N} \in (0,\infty)^{\N},
\end{equation}
and we let each pair of oscillators $\eta,\zeta\in\Omega_{N}$ at distance $d(\eta, \zeta) 
= d$ interact as in the mean-field Kuramoto model with $K/N$ replaced by $K_{d}/N^{2d-1}$, 
where the scaling factor is chosen to ensure that the model remains well behaved in the 
limit as $N \to\infty$. Thus, our coupled evolution equations read 
\begin{equation}
\label{eq:model}
\ddd\theta_{\eta}(t) = 
\sum_{\zeta\in \Omega_{N}} \frac{K_{d(\eta, \zeta)}}{N^{2d(\eta, \zeta)-1}}\,
\sin\big[\theta_{\zeta}(t) - \theta_{\eta}(t)\big]\,\ddd t + D\,\ddd W_{\eta}(t), 
\qquad \eta\in\Omega_{N}, t \geq 0,
\end{equation}
where $(W_{\eta}(t))_{t \geq 0}$, $\eta\in\Omega_{N}$, are i.i.d.\ standard Brownian motions.
As initial condition we take, as in \eqref{eq:initial},
\begin{equation}
\label{eq:initialhierar}
\begin{aligned}
\blacktriangleright\quad\,
&\theta_{\eta}(0),\,\eta \in \Omega_{N}, \text{ are i.i.d.\ and are drawn from}\\[-.1cm]  
&\text{a common probability distribution } \rho(\ddd\theta) \text{ on } [0,2\pi).
\end{aligned}
\end{equation}
We will be interested in understanding the evolution of average phase in the definition 
of the order parameter associated with the $N^{k}$ oscillators in the $k$-block around 
$\eta$ at time $N^kt$, defined by 
\begin{equation}
\label{eq:R}
R_{\eta, N}^{[k]}(Nt)\,\eee^{\im\Phi_{\eta, N}^{[k]}(t)} 
= \frac{1}{N^{k}} \sum_{\zeta \in B_{k}(\eta)} \eee^{\im\theta_{\zeta}(N^kt)}, 
\qquad \eta \in \Omega_{N}, t \geq 0, 
\end{equation}
where $R_{\eta,N}^{[k]}(Nt)$ is the synchronization level at time $N^kt$ and $\Phi_{\eta,N}^{[k]}(t)$ 
is the average phase at time $N^kt$. The new time scales $Nt$ and $t$ will turn out to be 
natural in view of the scaling in Theorem \ref{thm:slowangle}. The synchronization 
level $R_{\eta,N}^{[k]}$ captures the synchronization of the $(k-1)$-blocks, of which there are $N$ 
in total constituting the $k$-block around $\eta$. These blocks must synchronize before their average 
phase $\Phi_{\eta,N}^{[k]}$ can begin to move, which is why $R_{\eta,N}^{[k]}$ moves on a 
different time scale compared to $\Phi_{\eta,N}^{[k]}$. Our goal will be to pass to the limit $N\to\infty$, 
look at the limiting synchronization levels around a given vertex, say $\eta=0^\N$, and 
classify the scaling behavior of these synchronization levels as $k\to\infty$ into universality 
classes according to the choice of $(K_k)_{k\in\N}$ in \eqref{eq:Kdef}.

Note that, for every $\eta\in\Omega_N$, we can telescope to write
\begin{equation}
\label{eq:telescope}
\begin{aligned}
\sum_{\zeta \in \Omega_N} \frac{K_{d(\zeta, \eta)}}{N^{2d(\eta, \zeta)-1}}
\sin\big[\theta_{\zeta}(t) - \theta_{\eta}(t)\big] 
&= \sum_{k\in\N} \frac{K_{k}}{N^{2k-1}} \sum_{\zeta \in B_{k}(\eta)/B_{k-1}(\eta)}
\sin\big[\theta_{\zeta}(t) - \theta_{\eta}(t)\big]\\
&= \sum_{k\in\N} \Big(\frac{K_k}{N^{2k-1}} - \frac{K_{k+1}}{N^{2(k+1)-1}}\Big) 
\sum_{\zeta \in B_{k}(\eta)} \sin\big[\theta_{\zeta}(t) - \theta_{\eta}(t)\big].
\end{aligned}
\end{equation}
Inserting \eqref{eq:telescope} into \eqref{eq:model} and using \eqref{eq:R}, we get
\begin{equation}
\label{eq:remodel}
\begin{aligned}
&\ddd\theta_{\eta}(t)\\
&\qquad = \sum_{k\in\N} \frac{1}{N^{k-1}}\Big(K_{k} - \frac{K_{k+1}}{N^{2}}\Big)\,
R_{\eta, N}^{[k]}(N^{1-k}t)\,\sin\Big[\Phi_{\eta, N}^{[k]}(N^{-k}t) - \theta_{\eta}(t)\Big]\,\ddd t 
+ D\,\ddd W_{\eta}(t).
\end{aligned}
\end{equation}
This shows that, like in \eqref{eq:Knoiseorderpar}, \emph{the oscillators are coupled via the 
order parameters associated with the $k$-blocks for all $k\in\N$}, suitably weighted. 
As for the mean-field Kuramoto model, for every $\eta \in \Omega_N$, $R_{\eta, N}^{[k]}(N^{1-k}t)$
and $\Phi_{\eta, N}^{[k]}(N^{-k}t)$ are random variables that depend on $(K_k)_{k\in\N}$ and $D$.

When we pass to the limit $N\to\infty$ in \eqref{eq:remodel}, in the right-hand side of 
\eqref{eq:remodel} only the term with $k=1$ survives, so that we end up with an 
\emph{autonomous} evolution equation similar to \eqref{eq:propcha}. The goal of the 
present paper is to show that a similar decoupling occurs \emph{at all block levels}. 
Indeed, we expect the successive time scales at which synchronization occurs to separate. 
If there is synchronization at scale $k$, then we expect the average of the $k$-blocks 
around the origin forming the $(k+1)$-blocks (of which there are $N$ in total) to behave 
\emph{as if they were single oscillators} at scale $k+1$. 

Dahms~\cite{D02} considers a multi-layer model with a different type of interaction: single 
layers labelled by $\N$, each consisting of $N$ oscillators, are stacked on top of each other, 
and each oscillator in each layer is interacting with the \emph{average phases} of the oscillators 
in all the other layers, with interaction strengths $(\tilde K_k)_{k\in\N}$ (see \cite[Section 1.3]{D02}). 
For this model a necessary and sufficient criterion is derived for synchronization to be present 
at all levels in the limit as $N\to\infty$, namely, $\sum_{n\in\N} \tilde K_k^{-1} < \infty$ (see 
\cite[Section 1.4]{D02}). We will see that in our hierarchical model something similar is 
happening, but the criterion is rather more delicate.

%%%%%%%%%%%%%%% SECTION 2 %%%%%%%%%%%%%%%%%%%%%

\section{Main results}
\label{sec:results}

In Section~\ref{subsec:multi} we state a conjecture about the multi-scaling of the system 
(Conjecture~\ref{thm:scalphasend} below), which involves a renormalization transformation 
describing the synchronization level and the average phase on successive hierarchical levels. 
In Section~\ref{subsec:truncation} we propose a truncation approximation that simplifies the 
renormalization transformation, and argue why this approximation should be fairly accurate.  
In Section~\ref{subsec:universality} we analyse the simplified renormalization transformation 
and identify three universality classes for the behavior of the synchronization level as we 
move upwards in the hierarchy, give sufficient conditions on $(K_k)_{k\in\N}$ for each 
universality class (Theorem~\ref{thm:classesnd} below), and provide bounds on the 
synchronization level (Theorem~\ref{thm:critcasend} below). The details are given in 
Sections~\ref{sec:thmscalphase}--\ref{sec:thmclasses}. Without loss of generality we set 
$D=1$ in \eqref{eq:model}. 

%%%

\subsection{Multi-scaling}
\label{subsec:multi}

Our first result is a conjecture stating that the average phase of the $k$-blocks behaves like 
that of the noisy mean-field Kuramato model described in Theorem~\ref{thm:slowangle}. 
Recall the choice of time scales in \eqref{eq:R}. 

\begin{conjecture}{\bf (Multi-scaling for the block average phases)}
\label{thm:scalphasend}
Fix $k\in\N$ and assume that $R^{[k]}>0$. Then, in distribution,
\begin{equation}
\lim_{N\to\infty} \Phi_{0,N}^{[k]}(t) = \Phi_{0}^{[k]}(t),
\end{equation}
where $(\Phi_{0}^{[k]}(t))_{t \geq 0}$ evolves according to the SDE
\begin{equation}
\label{eq:thm:scalphasend}
\ddd\Phi_{0}^{[k]}(t) = K_{k+1}\,\mathcal{E}^{[k]}\,R_{0}^{[k+1]}(t)\,
\sin\big[\Phi - \Phi_{0}^{[k]}(t)\big]\,\ddd t + \mathcal{D}^{[k]}\,\ddd W_{0}^{[k]}(t),
\quad t \geq 0,
\end{equation}
$(W_{0}^{[k]}(t))_{t \geq 0}$ is a standard Brownian motion, $\Phi=0$ by calibration, 
and 
\begin{equation}
\label{eq:Tcompnd}
(\mathcal{E}^{[k]},\mathcal{D}^{[k]}) 
= \cT_{(K_\ell)_{1 \leq \ell \leq k}}(\mathcal{E}^{[0]},\mathcal{D}^{[0]}), \qquad k\in\N,
\end{equation}
with $(\mathcal{E}^{[0]},\mathcal{D}^{[0]}) = (1,1)$ and $\cT_{(K_\ell)_{1 \leq \ell \leq k}}$ a renormalization 
transformation.
\end{conjecture}

\noindent
The evolution in \eqref{eq:thm:scalphasend} is that of a mean-field noisy Kuramoto model 
with \emph{renormalized coefficients}, namely, an \emph{effective interaction strength} $K_{k+1}\,
\mathcal{E}^{[k]}$ and an \emph{effective noise strength} $\mathcal{D}^{[k]}$ (compare with 
\eqref{eq:Knoiseorderpar}). These coefficients are to be viewed as the result of a \emph{renormalization 
transformation} acting on block communities at levels $k \in \N$ successively, starting from the 
initial value $(\mathcal{E}^{[0]},\mathcal{D}^{[0]}) = (1,1)$. This initial value comes from the 
fact that single oscillators are completely synchronized by definition. The renormalization
transformation at level $k$ depends on the values of $K_\ell$ with $1 \leq \ell \leq k$. It also 
depends on the synchronization levels $R^{[\ell]}$ with $1 \leq \ell \leq k$, as well as on \emph{other 
order parameters} associated with the phase distributions of the $\ell$-blocks with $1 \leq \ell \leq k$.
In Section~\ref{subsec:truncation} we will analyse an \emph{approximation} for which this 
dependence simplifies, in the sense that only one set of extra order parameter comes into play, 
namely, $Q^{[\ell]}$ with $1 \leq \ell \leq k$, where $Q^{[\ell]}$ is the average of the cosine squared 
of the phase distribution of the $\ell$-block.      

The evolution in \eqref{eq:thm:scalphasend} is \emph{not closed} because of the presence 
of the term $R_{0}^{[k+1]}(t)$, which comes from the $(k+1)$-st block community one 
hierarchical level up from $k$. Similarly as in \eqref{eq:rrelax}, $R_{0}^{[k+1]}(t)$ is 
driven by a \emph{deterministic relaxation equation} such that 
\begin{equation}
\label{eq:Rrelax}
R_0^{[k+1]}(0)=R, \qquad \lim_{t\to\infty} R_0^{[k+1]}(t) = R^{[k+1]}. 
\end{equation}
This relaxation equation will be of no concern to us here (and is no doubt quite involved). 
Convergence holds at least for $R$ close to $R^{[k+1]}$ (recall Remark~\ref{rem:stability}). 
Thus, after a \emph{transient period}, \eqref{eq:thm:scalphasend} converges to the \emph{closed} 
evolution equation
\begin{equation}
\label{eq:thm:scalphasendalt}
\ddd\Phi_{0}^{[k]}(t) = K_{k+1}\,\mathcal{E}^{[k]}\,R^{[k+1]}\,
\sin\big[\Phi - \Phi_{0}^{[k]}(t)\big]\,\ddd t + \mathcal{D}^{[k]}\,\ddd W_{0}^{[k]}(t),
\quad t \geq 0.
\end{equation}
The initial values $(R,\Phi)$ in \eqref{eq:Rrelax} and \eqref{eq:thm:scalphasendalt} come 
from \eqref{RPhirho} and \eqref{eq:initialhierar}.

Conjecture~\ref{thm:scalphasend} perfectly fits the folklore of renormalization theory for 
interacting particle systems. The idea of that theory is that along an increasing sequence 
of mesoscopic space-time scales the evolution is the same as on the microscopic space-time 
scale, but with renormalised coefficients that arise from an `averaging out' on successive 
scales. It is generally hard to carry through a renormalization analysis in full detail, and 
there are only a handful of interacting particle systems for which this has been done with 
mathematical rigour. Moreover, there are delicate issues with the renormalization transformation 
being properly defined. However, in our model these issues should not arise because of the 
`layered structure' of the hierarchical lattice and the hierarchical interaction. Since the interaction 
between the oscillators is \emph{non-linear}, we currently have little hope to be able to turn 
Conjecture~\ref{thm:scalphasend} into a theorem and identify the precise form of 
$\cT_{(K_\ell)_{1 \leq \ell \leq k}}$. In Section~\ref{sec:multiscale} we will see that the 
non-linearity of the interaction causes a delicate interplay between the different hierarchical 
levels.     

In what follows we propose a simplified renormalization transformation $\bar\cT_{(K_\ell)_{1 
\leq \ell \leq k}}$, based on a \emph{truncation approximation} in which we keep only the 
interaction between \emph{successive} hierarchical levels. The latter \emph{can be analysed in 
detail} and replaces the renormalization transformation $\cT_{(K_\ell)_{1 \leq \ell \leq k}}$ 
in Conjecture~\ref{thm:scalphasend}, of which we do not know the details. We also argue 
why the truncation approximation is reasonable.           

%%%

\subsection{Truncation approximation}
\label{subsec:truncation}

The truncation approximation consists of replacing $\cT_{(K_\ell)_{1 \leq \ell \leq k}}$ by
a $k$-fold \emph{iteration of a renormalization map}:
\begin{equation}
\label{eq:composition}
\bar\cT_{(K_\ell)_{1 \leq \ell \leq k}} = \cT_{K_{k}} \circ \cdots \circ \cT_{K_1}.
\end{equation}  
In other words, we presume that what happens at hierarchical scale $k+1$ is dictated
only by what happens at hierarchical scale $k$, and not by any of the lower scales. 
These scales do manifest themselves via the successive interaction strengths, but 
not via a direct interaction.

Define
\begin{equation}
\label{eq:Zdef}
I_{0}(\lambda) = \frac{1}{2\pi} \int_0^{2\pi} \ddd\phi\,\,\eee^{\lambda\cos\phi}, 
\qquad \lambda>0,
\end{equation}
which is the modified Bessel function of the first kind.  After normalization, the integrand becomes 
what is called the von Mises probability density function on the unit circle with parameter 
$\lambda$, which is $\phi \mapsto p_\lambda(\phi,0)$ in \eqref{eq:stat}--\eqref{eq:ABdef}. 
We write $I_{0}'(\lambda) = I_{1}(\lambda)$ and $I_{0}''(\lambda) = I_{2}(\lambda)$.

\begin{definition}{\bf (Renormalization map)}
\label{def:TKnd} 
{\rm For $K \in (0,\infty)$, let $\cT_K\colon\,[0,1] \times [\tfrac12,1] \to [0,1] \times [\tfrac12,1]$ 
be the map 
\begin{equation}
\label{eq:TKmap}
(R',Q') = \cT_K(R,Q)
\end{equation} 
defined by
\begin{equation}
\label{eq:RQbar}
\begin{aligned}
R' &= R\,\frac{I_{1}(2KR'\sqrt{Q})}{I_{0}(2KR'\sqrt{Q})},\\
Q'-\tfrac12 &= (Q-\tfrac12) \,\Bigg[2\,\frac{I_{2}(2KR'\sqrt{Q})}{I_{0}(2KR'\sqrt{Q})}-1\Bigg].
\end{aligned}
\end{equation}
The first equation is a \emph{consistency relation}, the second equation is a 
\emph{recursion relation}. They must be used in that order to find the image point
$(R',Q')$ of the original point $(R,Q)$ under the map $\cT_K$.} \hfill\qed
\end{definition}  

With this renormalization mapping we can approximate the true renormalized system.

\begin{approximation}
After truncation, \eqref{eq:thm:scalphasend} can be approximated by 
\begin{equation}
\label{eq:thm:scalphasendapprox}
\ddd\Phi_{0}^{[k]}(t) = K_{k+1}\,\bar{\mathcal{E}}^{[k]}\,R_{0}^{[k+1]}(t)\,
\sin\big[\Phi - \Phi_{0}^{[k]}(t)\big]\,\ddd t + \bar{\mathcal{D}}^{[k]}\,\ddd W_{0}^{[k]}(t),
\quad t \geq 0,
\end{equation}
with
\begin{equation}
\label{eq:EDdef}
\bar{\mathcal{E}}^{[k]} = \frac{Q^{[k]}}{R^{[k]}}, \qquad 
\bar{\mathcal{D}}^{[k]} = \frac{\sqrt{Q^{[k]}}}{R^{[k]}},
\end{equation}
where
\begin{equation}
(R^{[k]},Q^{[k]}) = \bar\cT_{(K_\ell)_{1 \leq \ell \leq k}}(R^{[0]},Q^{[0]}),
\qquad (R^{[0]},Q^{[0]}) = (1,1).
\end{equation}
\end{approximation}   

\noindent
We will see in Section~\ref{sec:multiscale} that $R^{[k]}$ plays the role of the synchronization 
level of the $k$-blocks, while $Q^{[k]}$ plays the role of the average of the cosine squared of 
the phase distribution of the $k$-blocks (see \eqref{eq:barRQdef} below).  

In the remainder of this section we analyse the orbit $k \mapsto (R^{[k]},Q^{[k]})$ in detail. 
We will see that, under the simplified renormalization transformation, $k \mapsto (R^{[k]},Q^{[k]})$ 
is non-increasing in both components. In particular, synchronization cannot increase when 
the hierarchical level goes up. 

\begin{remark}
{\rm In Section~\ref{sec:multiscale} we will argue that a better approximation can be obtained 
by keeping one more term in the truncation approximation, but that the improvement is minor.} 
\hfill\qed
\end{remark}

%%%

\subsection{Universality classes}
\label{subsec:universality}

There are \emph{three universality classes} depending on the choice of $(K_k)_{k\in\N}$ 
in \eqref{eq:Kdef}, illustrated in Fig.~\ref{fig:rqcomic}:

%%%%%%%%%%%%%%%%%%%%%%%%%%%%%%%%%
\begin{figure}[htbp]
\centering
\includegraphics[scale=0.3]{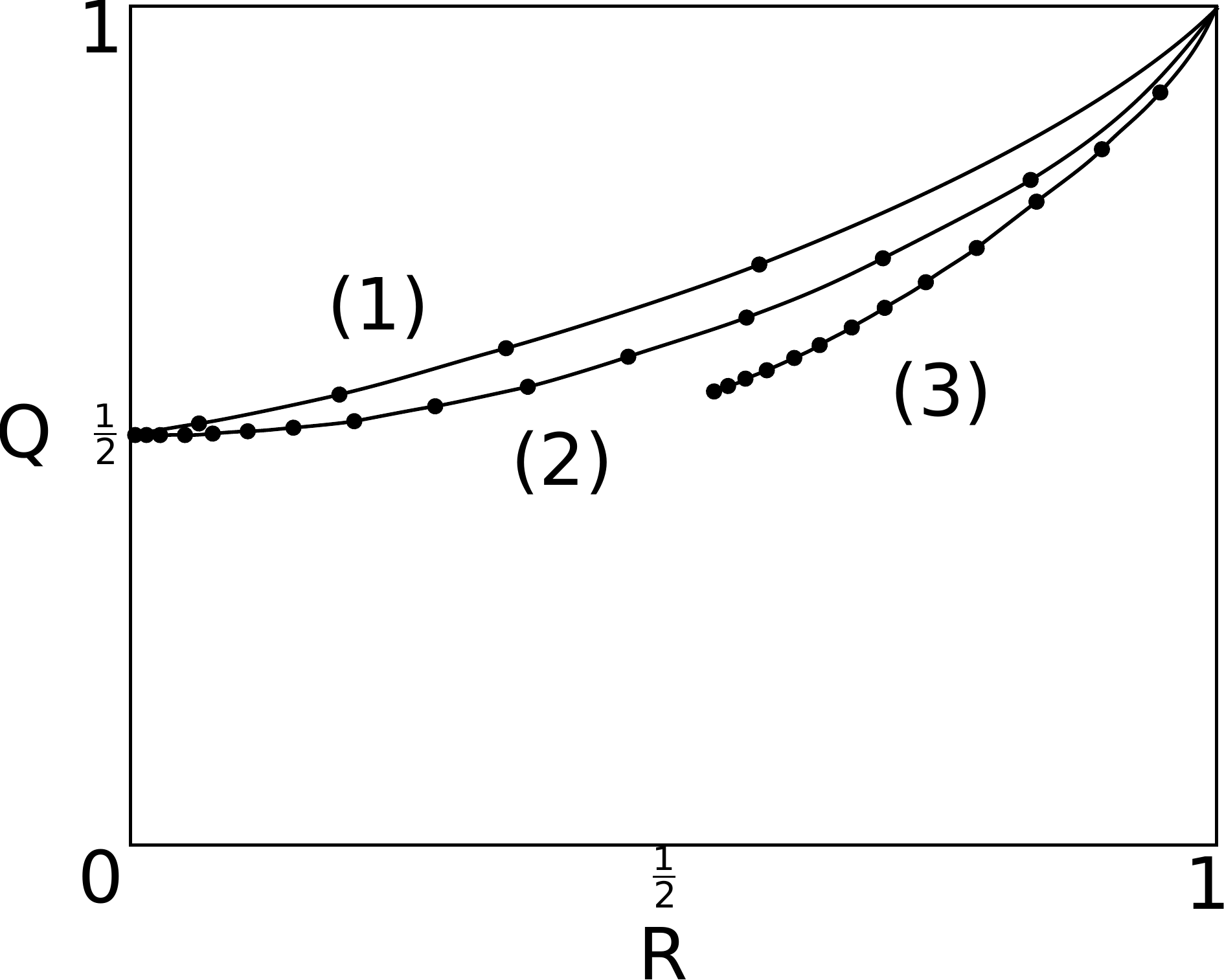}
\caption{The dots represent the map $k \mapsto (R^{[k]},Q^{[k]})$ for the three universality 
classes, starting from $(R^{[0]},Q^{[0]})=(1,1)$. The dots move left and down as $k$ increases.}
\label{fig:rqcomic}
\end{figure} 
%%%%%%%%%%%%%%%%%%%%%%%%%%%%%%%%%%

\begin{enumerate}
\item[(1)] 
Synchronization is lost at a finite level:\\ 
$R^{[k]}>0$, $0 \leq k < k_*$, $R^{[k]}=0$, $k \geq k_*$ for some $k_*\in\N$. 
\item[(2)] 
Synchronization is lost asymptotically:\\ 
$R^{[k]}>0$, $k \in \N_0$, $\displaystyle\lim_{k\to\infty} R^{[k]}=0$.
\item[(3)] 
Synchronization is not lost asymptotically:\\
$R^{[k]}>0$, $k \in \N_0$, $\displaystyle\lim_{k\to\infty} R^{[k]}>0$.
\end{enumerate}

\noindent
Our second result provides sufficient conditions for universality classes (1) and (3)
in terms of the sum $\sum_{k\in\N} K_{k}^{-1}$. 

\begin{theorem}{\bf (Criteria for the universality classes)}
\label{thm:classesnd}
$\mbox{}$
\begin{itemize}
\item
$\sum_{k\in\N} K_{k}^{-1} \geq 4$ $\Longrightarrow$ universality class {\rm (1)}.
\item
$\sum_{k\in\N} K_{k}^{-1} \leq \frac{1}{\sqrt{2}}$ $\Longrightarrow$ universality class {\rm (3)}.
\qed
\end{itemize}
\end{theorem}

\noindent
Two examples are: (1) $K_{k} = \frac{3}{2\log 2}\log(k+1)$; (3) $K_{k} = 4e^{k}$. 
The scaling behaviour for these examples is illustrated via the numerical analysis in 
Appendix~\ref{app:numerics} (see, in particular, Fig.~\ref{fig:uniclass1} and Fig.~\ref{fig:uniclass3}
below).

The criteria in Theorem~\ref{thm:classesnd} are not sharp. Universality class (2) corresponds 
to a \emph{critical surface} in the space of parameters $(K_{k})_{k \in N}$ that appears to 
be rather complicated and certainly is not (!) of the type $\sum_{k\in\N} K_{k}^{-1} = c$ for 
some $\frac{1}{\sqrt{2}}<c<4$ (see Fig.~\ref{fig:uniclassescomic}). Note that the full sequence 
$(K_k)_{k\in\N}$ determines in which universality class the system is. 

%%%%%%%%%%%%%%%%%%%%%%%%%%%%%%%%%%%%%%%%%
\begin{figure}[htbp]
\vspace{0.3cm}
\centering
\includegraphics[scale=0.3]{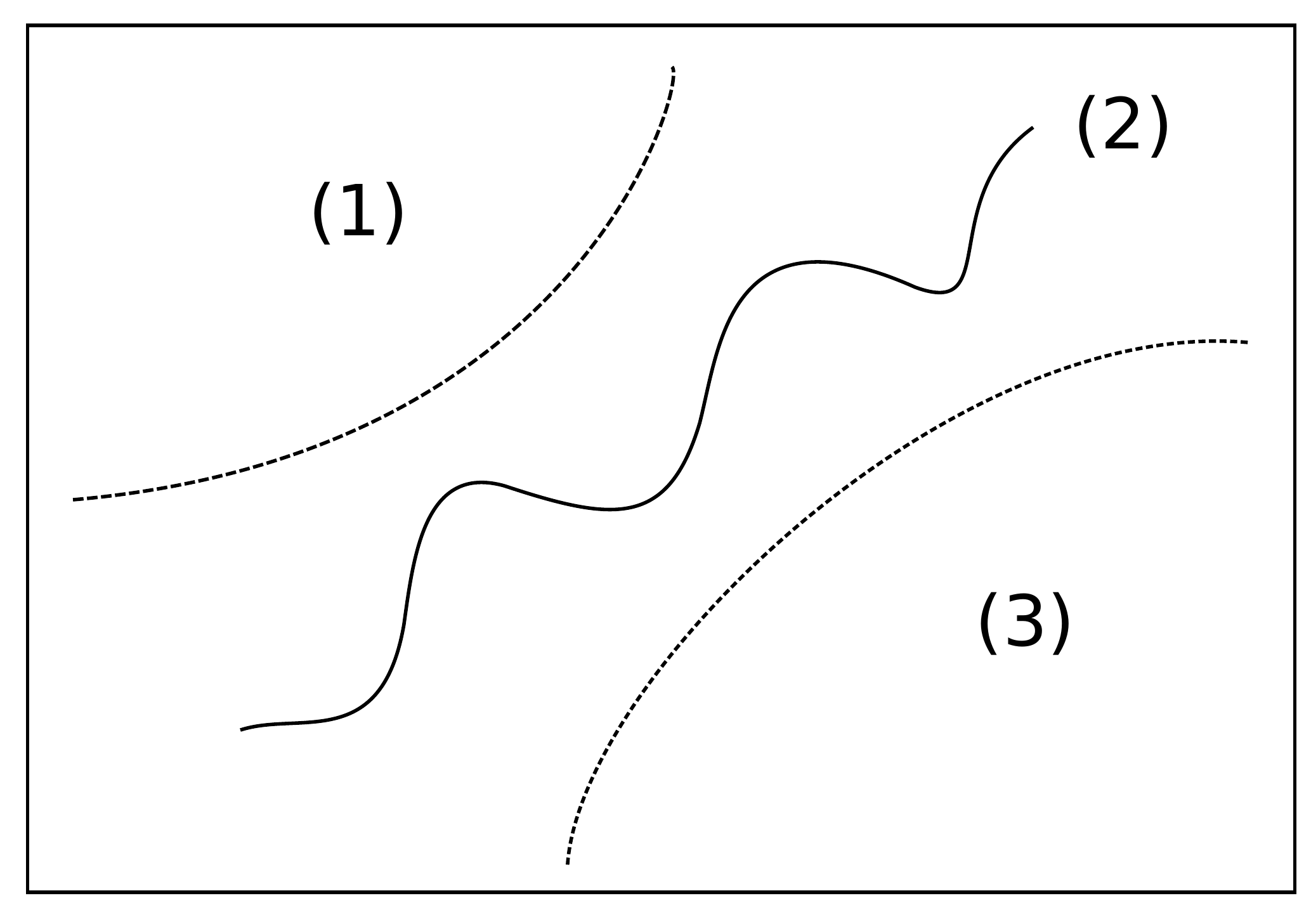}
\caption{Caricature showing the critical surface in the parameter space and the bounds 
provided by Theorem~\ref{thm:classesnd}.}
\label{fig:uniclassescomic}
\end{figure} 
%%%%%%%%%%%%%%%%%%%%%%%%%%%%%%%%%%%%%%%%%%

The behaviour of $K_k$ as $k\to\infty$ determines the speed at which $R^{[k]} \to R^{[\infty]}$ 
in universality classes (2) and (3). Our third theorem provides upper and lower bounds.

\begin{theorem}{\bf (Bounds for the block synchronization levels)}
\label{thm:critcasend} 
\begin{itemize}
\item
In universality classes {\rm (2)} and {\rm (3)}, 
\begin{equation}
\label{eq:synbds}
\tfrac14 \sigma_k \leq R^{[k]} - R^{[\infty]} \leq \sqrt{2}\,\sigma_k, \qquad k \in \N_0,
\end{equation}
with $\sigma_k = \sum_{\ell>k} K_{\ell}^{-1}$.
\item 
In universality class {\rm (1)}, the upper bound in \eqref{eq:synbds} holds for $k\in\N_0$, while 
the lower bound in \eqref{eq:synbds} is replaced by
\begin{equation}
\label{eq:univ1lb1}
R^{[k]} - R^{[k_*-1]} \geq \tfrac14 \sum_{\ell=k+1}^{k_*-1} K_{\ell}^{-1}, \qquad 0 \leq k \leq k_*-2.
\end{equation}
The latter implies that
\begin{equation}
\label{eq:univ1lb2}
k_* \leq \max\left\{k \in \N\colon\, \sum_{\ell=1}^{k-1} K_\ell^{-1} < 4\right\}
\end{equation}
because $R^{[0]}=1$ and $R^{[k_*-1]}>0$.
\end{itemize}
\end{theorem}

\noindent
In universality classes (2) and (3) we have $\lim_{k\to\infty} \sigma_k=0$. Depending on how 
fast $k \mapsto K_k$ grows, various speeds of convergence are possible: logarithmic, polynomial, 
exponential, superexponential.

%%%%%%%%%%%%%%%% SECTION 3 %%%%%%%%%%%%%%%%%%%%%%%%%%%

\section{Multi-scaling for the block average phases}
\label{sec:thmscalphase}

In Section~\ref{sec:prep} we explain the heuristics behind Theorem~\ref{thm:slowangle}. The diffusive 
scaling of the average phase in the mean-field noisy Kuramato model, as shown in the first line of 
\eqref{eq:rpsitimescale}, is a key tool in our analysis of the multi-scaling of the block average 
phases in the hierarchical noisy Kuramoto model, stated in Conjecture~\ref{thm:scalphasend}. The 
justification for the latter is given in Section~\ref{sec:multiscale}.

%%%

\subsection{Diffusive scaling of the average phase for mean-field Kuramato}
\label{sec:prep}

\begin{proof}
For the heuristic derivation of the second line of \eqref{eq:rpsitimescale} we combine 
\eqref{eq:McKean}--\eqref{eq:orderpar}, to obtain
\begin{equation}
\begin{aligned}
\frac{\ddd}{\ddd t} r(t) &= \int_0^{2\pi} \ddd\theta\,\cos\theta\\
&\qquad \times \left\{-\frac{\partial}{\partial \theta}\,\Big[p_\lambda(t;\theta)\,
\big\{Kr(t)\sin[\psi(t) - \theta]\big\}\Big] 
+ \frac{1}{2}\frac{\partial^{2}}{\partial \theta^{2}}\,p_\lambda(t;\theta)\right\}
\end{aligned}
\end{equation}
with $\lambda=2Kr$ and $p_\lambda(t;\theta)=p_\lambda(t;\theta,0)$ (recall that $\omega
\equiv 0$). After partial integration with respect to $\theta$ this becomes (use that $\theta 
\mapsto p_\lambda(t;\theta)$ is periodic)
\begin{equation}
\label{eq:drdt}
\frac{\ddd}{\ddd t} r(t) = \int_0^{2\pi} \ddd\theta\,
p_\lambda(t;\theta)\,\left\{(-\sin\theta)\,Kr(t)\sin(-\theta)
+ (-\cos\theta)\,\frac{1}{2}\right\},
\end{equation}
where we use that $\psi(t)=\psi(0)=0$. Define
\begin{equation}
\label{qdef}
q(t) = \int_0^{2\pi} \ddd\theta\,p_\lambda(t;\theta)\,\cos^2\theta.
\end{equation}
Then \eqref{eq:drdt} reads
\begin{equation}
\frac{\ddd}{\ddd t} r(t) =  \left[K(1-q(t))-\frac{1}{2}\right]\,r(t).
\end{equation}
We know that
\begin{equation}
\lim_{t\to\infty} q(t) = q = \int_0^{2\pi} \ddd\theta\,p_\lambda(\theta)\,\cos^2\theta
\end{equation} 
with (put $\omega \equiv 0$ in \eqref{eq:stat})
\begin{equation}
\label{eq:plt}
p_\lambda(\theta) = \frac{\eee^{\lambda\cos\theta}}{\int_0^{2\pi} \ddd\phi\,\eee^{\lambda\cos\phi}}.
\end{equation}
Note that $K(1-q)-\frac{1}{2}=0$ because $\lambda=2Kr$ and
\begin{equation}
\label{eq:pltalt}
\int_0^{2\pi} \ddd\theta\, p_\lambda (\theta) \sin^2\theta 
= (1/\lambda) \int_0^{2\pi} \ddd\theta\, p_\lambda(\theta) \cos\theta = r/\lambda
\end{equation} 
by partial integration. Hence $\lim_{t\to\infty} r(t)=r$.  (The  fine details of the relaxation are delicate, 
depend on the full solution of the McKean-Vlasov equation in \eqref{eq:McKean}, but are of no concern 
to us here.)
 
For the derivation of the first line of \eqref{eq:rpsitimescale} we use the symmetry of the equilibrium 
distribution (recall \eqref{eq:stat}--\eqref{eq:ABdef}), i.e.,
\begin{equation}
\label{eq:preequisym}
p_\lambda(\theta) = p_\lambda(-\theta),
\end{equation} 
together with the fact that $x \mapsto \cos x$ is a symmetric function and $x \mapsto \sin x$ is 
an asymmetric function.
 
Write the definition of the order parameter as
\begin{equation}
\label{rpsirel}
r_N = \frac{1}{N} \sum_{j=1}^N \eee^{\im(\theta_j-\psi_N)} 
\end{equation}
and compute
\begin{equation}
\frac{\partial r_N}{\partial\theta_i} = \frac{\im}{N}\,\eee^{\im(\theta_i-\psi_N)} 
- \im\,\frac{\partial\psi_N}{\partial \theta_k}\,r_N.
\end{equation}
Collecting the real and the imaginary part, we get
\begin{equation}
\frac{\partial \psi_N}{\partial\theta_i} = \frac{1}{Nr_N} \cos(\psi_N-\theta_i),
\qquad \frac{\partial r_N}{\partial\theta_i} = \frac{1}{N} \sin(\psi_N-\theta_i).
\end{equation}
One further differentiation gives
\begin{equation}
\begin{aligned}
\frac{\partial^2 \psi_N}{\partial\theta_i^2}
&= - \frac{1}{Nr_N^2} \frac{\partial r_N}{\partial\theta_i} \cos(\psi_N-\theta_i)
- \frac{1}{Nr_N} \left[\frac{\partial \psi_N}{\partial\theta_i}-1\right] \cos(\psi_N-\theta_i)\\
&= - \frac{2}{(Nr_N)^2} \sin(\psi_N-\theta_i) \cos(\psi_N-\theta_i)
+ \frac{1}{Nr_N} \sin(\psi_N-\theta_i),
\end{aligned} 
\end{equation}
plus a similar formula for $\frac{\partial^2 r_N}{\partial\theta_i^2}$ (which we will not need).
Thus, It\^o's rule applied to \eqref{eq:orderparfiniteN} yields the expression
\begin{equation}
\label{eq:preito}
\ddd\psi_N(t) 
= \sum_{i=1}^N \frac{\partial\psi_N}{\partial \theta_i}(t)\,\ddd\theta_i(t) 
+ \frac{1}{2} \sum_{i=1}^N \frac{\partial^2\psi_N}{\partial \theta_i^2}(t)\,\big(\ddd\theta_i(t)\big)^2 
\end{equation} 
with 
\begin{eqnarray}
\label{eq:preito1}
\frac{\partial\psi_N}{\partial \theta_i}(t)\,
&=& \frac{1}{Nr_N(t)} \cos\big[\psi_N(t) - \theta_i(t)\big],\\ \nonumber
\label{eq:preito2}
\frac{\partial^2\psi_N}{\partial \theta_i^2}(t)
&=& - \frac{2}{\big(Nr_N(t))^2} \sin\big[\psi_N(t) - \theta_i(t)\big]
\cos\big[\psi_N(t) - \theta_i(t)\big]\\ \nonumber
&& \qquad + \frac{1}{Nr_N(t)} \sin\big[\psi_N(t) - \theta_i(t)\big].
\end{eqnarray}
Inserting \eqref{eq:Knoiseorderpar} into \eqref{eq:preito}--\eqref{eq:preito2}, we get
\begin{equation}
\label{eq:prebeforelimit1}
\ddd\psi_N(t) = I(N;t)\,\ddd t + \ddd J(N;t)
\end{equation}
with
\begin{equation}
\label{eq:prebeforelimit2}
\begin{aligned}
I(N;t)
&= \left[\frac{K}{N}-\frac{1}{\big(N r_N(t)\big)^2}\right] 
\sum_{i=1}^N \sin\big[\psi_N(t) - \theta_i(t)\big] \cos\big[\psi_N(t) - \theta_i(t)\big],\\
\ddd J(N;t)
&= \frac{1}{Nr_N(t)} \sum_{i=1}^N \cos\big[\psi_N(t) - \theta_i(t)\big]\,\ddd W_i(t),
\end{aligned}
\end{equation}
where we use that $\sum_{i=1}^N \sin[\psi_N(t) - \theta_i(t)]=0$ by \eqref{eq:orderparfiniteN}. 
Multiply time by $N$, to get
\begin{equation}
\label{eq:prebeforelimit1alt}
\ddd\psi_N(Nt) = NI(N;Nt)\,\ddd t + \ddd J(N;Nt)
\end{equation}
with
\begin{equation}
\label{eq:prebeforelimit2alt}
\begin{aligned}
NI(N;Nt)
&= \left[K-\frac{1}{N\big(r_N(Nt)\big)^2}\right] 
\sum_{i=1}^N \sin\big[\psi_N(Nt) - \theta_i(Nt)\big] \cos\big[\psi_N(Nt) - \theta_i(Nt)\big],\\
\ddd J(N;Nt)
&= \frac{1}{Nr_N(Nt)} \sum_{i=1}^N \cos\big[\psi_N(Nt) - \theta_i(t)\big]\,\ddd W_i(Nt).
\end{aligned}
\end{equation}

Suppose that the system converges to a partially synchronized state, i.e., in distribution
\begin{equation}
\label{eq:rNNtalt}
\lim_{N\to\infty} r_N(Nt) = r > 0 \qquad \forall\,t>0
\end{equation}
(recall \eqref{eq:rNNt}). Then $\lim_{N\to\infty} 1/N(r_N(Nt))^2=0$, and so the first line in 
\eqref{eq:prebeforelimit2alt} scales like
\begin{equation}
\label{eq:maindiff*}
K \sum_{i=1}^N \sin\big[\psi_N(Nt) - \theta_i(Nt)\big] \cos\big[\psi_N(Nt) - \theta_i(Nt)\big],
\qquad N \to \infty.
\end{equation}
This expression is a large sum of terms whose average with respect to the noise is close to 
zero because of \eqref{eq:preequisym}. Consequently, this sum behaves diffusively. Also the 
second line in \eqref{eq:prebeforelimit2alt} behaves diffusively, because it is equal in distribution 
to
\begin{equation}
\label{eq:maindiff}
\frac{1}{r_N(Nt)} \sqrt{\frac{1}{N} \sum_{i=1}^N \cos^2\big[\psi_N(Nt) - \theta_i(Nt)\big]} 
\,\,\ddd W_*(t).
\end{equation}
It is shown in \cite{BGP14} that the two terms \emph{together} lead to the first line 
of \eqref{eq:rpsitimescale}, i.e., in distribution
\begin{equation}
\label{eq:psiNNt}
\lim_{N\to\infty} \psi_N(Nt) = \psi_*(t)
\end{equation}
with
\begin{equation} 
\psi_*(t) = D_*\,W_*(t), \qquad \psi_*(0) = \Phi = 0,
\end{equation} 
where $D_*=D_*(K)$ is the \emph{renormalized noise strength} given by \eqref{D*def} 
with $D=1$.\footnote{The proof is based on Hilbert-space techniques and is delicate. As pointed 
out below \cite[Corollary 1.3]{BGP14}: the proof requires control of the evolution of the empirical 
distribution of the oscillators, and so \eqref{eq:prebeforelimit1}--\eqref{eq:prebeforelimit2} 
alone cannot provide an alternative route to the estimates that are needed to prove the
convergence and the persistence of proximity in \eqref{eq:rNNtalt} and \eqref{eq:psiNNt}.}

Note that the term under the square root in \eqref{eq:maindiff} converges to $q$ 
defined in \eqref{qdef}. The latter holds because $\theta_i(Nt)$, $i=1,\ldots,N$, are 
asymptotically independent and $\theta_i(Nt)$ converges in distribution to $\theta \mapsto 
p_{\lambda}(\theta)$ \emph{relative} to the value of $\psi_N(Nt)$, which itself evolves 
only \emph{slowly} (on time scale $Nt$ rather than $t$).
\end{proof}

%%%%%%%%%%%%%%%%%%%%%%%%%%%%%%%%%%%%%%%%%
\begin{figure}[htbp]
\vspace{0.3cm}
\centering
\includegraphics[scale=1.0]{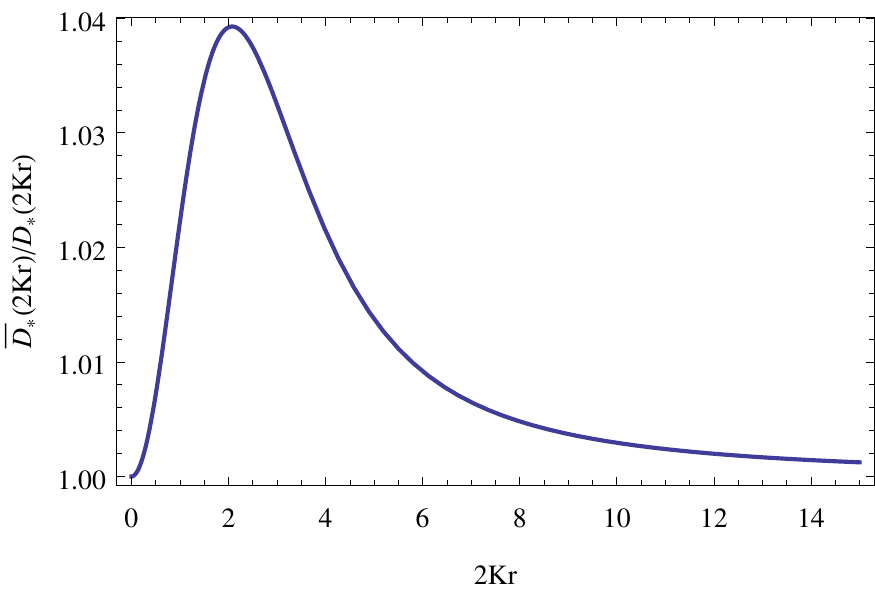}
\caption{Plot of $\bar D_*/D_*$ as a function of $2Kr$.}
\label{fig:twoD}
\end{figure} 
%%%%%%%%%%%%%%%%%%%%%%%%%%%%%%%%%%%%%%%%%%

The second line of \eqref{eq:prebeforelimit2alt} scales in distribution to the diffusion equation
\begin{equation}
\label{eq:qrdef}
\lim_{N\to\infty} \ddd J(N;Nt) = \bar D_* \ddd W_*(t), 
\qquad \bar D_* = D_*(K) = \frac{\sqrt{q}}{r}, \qquad r=r(K).
\end{equation}
Inserting \eqref{eq:plt} and recalling \eqref{eq:Zdef} and \eqref{qdef}, we have
\begin{equation}
\label{barD*def}
\bar D_* = \bar D_*(K) = \frac{1}{r}\,\sqrt{\frac{I_2(2Kr)}{I_{0}(2Kr)}}.
\end{equation} 
Clearly, 
$D_* \neq \bar D_*$. Interestingly, however,
\begin{equation}
\label{ratioD}
1 \leq \frac{\bar D_*}{D_*} \leq C \quad \text{ uniformly in } K \quad \text{ with } C = 1.0392 \dots
\end{equation}
(G.\ Giacomin, private communication). Hence, not only does the first line of \eqref{eq:prebeforelimit2alt}
\emph{lower} the diffusion constant, the amount by which it does so is less than 4 percent (see 
Fig.~\ref{fig:twoD}). Further thoughts on the reason behind the discrepancy between $D_*$ 
and $\bar D_*$ can be found in Dahms~\cite[Section 3.5]{D02}.

%%%

\subsection{Multi-scaling of the block average phases for hierarchical Kuramoto}
\label{sec:multiscale}

We give the main idea behind Conjecture~\ref{thm:scalphasend}. The argument runs along 
the same line as in Section~\ref{sec:prep}, but is more involved because of the hierarchical 
interaction. 

What is crucial for the argument is the \emph{separation of space-time scales}: 
\begin{itemize}
\item
Each $k$-block consists of $N$ disjoint $(k-1)$-blocks, and evolves on a time scale that is 
$N$ times larger than the time scale on which the constituent blocks evolve. 
\item
In the limit as $N\to\infty$, the constituent $(k-1)$-blocks in each $k$-block rapidly achieve equilibrium
subject to the \emph{current} value of the $k$-block, which allows us to treat the $k$-blocks as a noisy 
mean-field Kuramoto model with coefficients that depend on their internal synchronization level, with an 
effective interaction that depends on the current value of the synchronization level of the $(k+1)$-block.
\item 
The $k$-block itself interacts with the other $k$-block's, with interaction strength 
$K_{k+1}$, while the interaction with the even larger blocks it is part of is negligible as 
$N\to\infty$. This interaction occurs through an effective interaction with the average value 
of the $k$-blocks which is exactly the value of the $(k+1)$-block.
\end{itemize}

If we want to observe the evolution of the $k$-blocks labeled $1 \leq i \leq N$ that make up the 
$(k+1)$-block (i.e., the $\Phi_{i}^{[k]}(t)$'s) on time scale $t$), then we must scale the actual 
oscillator time by $N^{k}$. The synchronization levels within the $\Phi_{i}^{[k]}(t)$'s, given by 
$R_{i}^{[k]}(Nt)$, are then moving over time $Nt$, since they must be synchronized before 
the $\Phi_{i}^{[k]}(t)$'s start to diffuse. This is taken care of by our choice of time scales in the 
hierarchical order parameter \eqref{eq:R}.

It\^o's rule applied to \eqref{eq:R} with $\eta=0^\N$ gives
\begin{equation}
\label{eq:ito}
\ddd\Phi_{0}^{[k]}(t) 
= \sum_{\zeta\in B_{k}(0)} \frac{\partial\Phi_{0}^{[k]}}{\partial \theta_{\zeta}}(t)\,
\ddd\theta_{\zeta}(N^{k}t) + \frac{1}{2} \sum_{\zeta \in B_{k}(0)} 
\frac{\partial^{2}\Phi_{0}^{[k]}}{\partial \theta_{\zeta}^{2}}(t)\,\big(\ddd\theta_{\zeta}(N^{k}t)\big)^2
\end{equation} 
where we have suppressed the $N$-dependence in order to lighten the notation, writing  
$\Phi_{0}^{[k]}=\Phi_{0, N}^{[k]}$ and $R_{0}^{[k]}=R_{0, N}^{[k]}$. The derivatives are 
(compare with \eqref{eq:preito1})
\begin{eqnarray}
\label{eq:ito1}
\frac{\partial \Phi_{0}^{[k]}}{\partial \theta_{\zeta}}(t)
&=& \frac{1}{N^{k}R_{0}^{[k]}(Nt)}\cos\big[\Phi_{0}^{[k]}(t) - \theta_{\zeta}(N^kt)\big],\\
\label{eq:ito2}
\frac{\partial^{2} \Phi_{0}^{[k]}}{\partial \theta_{\zeta}^{2}}(t) 
&=& - \frac{2}{\big[N^{2k}R_{0}^{[k]}(Nt)\big]^{2}}\sin\big[\Phi_{0}^{[k]}(t) - \theta_{\zeta}(N^kt)\big]
\cos\big[\Phi_{0}^{[k]}(t) - \theta_{\zeta}(N^kt)\big]\\ \nonumber
&&\qquad + \frac{1}{N^{k}R_{0}^{[k]}(Nt)}\sin\big[\Phi_{0}^{[k]}(t) - \theta_{\zeta}(N^kt)\big].
\end{eqnarray}
Inserting \eqref{eq:remodel}, we find
\begin{equation}
\label{eq:beforelimit1}
d\Phi_{0}^{[k]}(t) = \big[I_1(k,N;t) + I_2(k,N;t)\big]\, \ddd t + \ddd J(k,N;t)
\end{equation}
with
\begin{equation}
\label{eq:beforelimit2}
\begin{aligned}
I_1(k,N;t) &= \frac{1}{R_{0}^{[k]}(Nt)} \sum_{\ell\in\N} \frac{1}{N^{\ell-1}}\,
\Big(K_{\ell} - \frac{K_{\ell+1}}{N^{2}}\Big)\\
&\hspace{-1cm}\times \sum_{\zeta\in B_{k}(0)}
R_{\zeta}^{[\ell]}(N^{1+k-\ell}t) \sin\big[\Phi_{\zeta}^{[\ell]}(N^{k-\ell}t) - \theta_{\zeta}(N^{k}t)\big]
\cos\big[\Phi_{0}^{[k]}(t) - \theta_{\zeta}(N^{k}t)\big],\\
I_2(k,N;t) &= - \frac{1}{N^k\big[R_{0}^{[k]}(Nt)\big]^{2}} \sum_{\zeta\in B_{k}(0)}
\sin\big[\Phi_{0}^{[k]}(t) - \theta_{\zeta}(N^{k}t)\big] \cos\big[\Phi_{0}^{[k]}(t) - \theta_{\zeta}(N^{k}t)\big],\\
\ddd J(k,N;t) &= \frac{1}{N^{k/2}R_{0}^{[k]}(Nt)} \sum_{\zeta\in B_{k}(0)}
\cos\big[\Phi_{0}^{[k]}(t) - \theta_{\zeta}(N^{k}t)\big]\,\ddd W_{\zeta}(t).
\end{aligned}
\end{equation}

Our goal is to analyse the expressions in \eqref{eq:beforelimit2} in the limit as $N\to\infty$, and 
show that \eqref{eq:beforelimit1} converges to the SDE in \eqref{eq:thm:scalphasend} subject
to the assumption that the $k$-block converges to a partially synchronized state, i.e.,
\begin{equation}
\lim_{N\to\infty} R_0^{[k]}(Nt) = R^{[k]} > 0 \quad \forall \; t>0. 
\end{equation}
The key idea is that, in the limit as $N\to\infty$, the average phases of the $k$-blocks around $\zeta$
\emph{decouple} and converge in distribution to $\theta \mapsto p^{[k]}(\theta)$ for all $k\in\N_0$, 
just as for the noisy mean-field Kuramoto model discussed in Section~\ref{sec:prep}, with $p^{[k]}(\theta)$
of the same form as $p_\lambda(\theta)$ in \eqref{eq:plt} for a suitable $\lambda$ depending on $k$. 
This is the reason why a \emph{recursive structure} is in place, captured by the renormalization maps 
$\cT_{K_k}$, $k\in\N$. 

Along the way we need the quantities
\begin{equation}
\label{eq:barRQdef}
\begin{aligned}
R_{0}^{[k]}(Nt) &= \frac{1}{N^{k}} 
\sum_{\zeta \in B_{k}(0)} \cos\big[\Phi_{0}^{[k]}(t) - \theta_{\zeta}(N^{k}t)\big],\\
Q_{0}^{[k]}(Nt) &= \frac{1}{N^{k}} 
\sum_{\zeta \in B_{k}(0)} \cos^2\big[\Phi_{0}^{[k]}(t) - \theta_{\zeta}(N^{k}t)\big].
\end{aligned}
\end{equation}
We also use that for all $k \in \N_0$,
\begin{equation}
\label{eq:equisym}
p^{[k]}(\theta) = p^{[k]}(-\theta),
\end{equation}
as well as the fact that for all $k \in \N$ and $\ell \geq k$,
\begin{equation}
\label{eq:embedding}
\begin{array}{l}
R_{\zeta}^{[\ell]}(Nt)=R_{0}^{[\ell]}(Nt),\\
\Phi_{\zeta}^{[\ell]}(Nt) = \Phi_{0}^{[\ell]}(Nt),
\end{array}
\qquad \forall \, \zeta \in B_k(0).
\end{equation}
In addition, we use the trigonometric identities
\begin{equation}
\label{trig}
\begin{array}{ll}
\sin(a+b) &= \sin a\cos b +\cos a \sin b,\\
\cos(a+b) &= \cos a\cos b - \sin a \sin b,
\end{array}
\qquad a,b \in \R,
\end{equation} 
to simplify terms via a \emph{telescoping argument}. 

Before we embark on our multi-scale analysis, we note that the expressions in 
\eqref{eq:beforelimit1}--\eqref{eq:beforelimit2} simplify somewhat as we take the limit $N\to\infty$. 
First, in $I_1(k,N;t)$ the term $K_{\ell+1}/N^{2}$ is asymptotically negligible compared to 
$K_\ell$, while the sum over $\ell$ can be restricted to $1 \leq \ell \leq k+1$ because $|B_k(0)| 
= N^k$. Second, $I_2(k,N;t)$ is asymptotically negligible because of \eqref{eq:equisym} and 
the fact that $\sin\theta\cos\theta = - \sin(-\theta)\cos(-\theta)$. Thus, we have, in distribution,
 \begin{equation}
\label{eq:beforelimitsimplified}
\ddd\Phi_{0}^{[k]}(t) = \big\{[1+o(1)]\,I_{N}^{[k]}(t)+o(1)\big\}\, \ddd t + \ddd J_{N}^{[k]}(t), \qquad N\to\infty,
\end{equation}
with
\begin{equation}
\begin{aligned}
\label{IJkNt}
I_{N}^{[k]}(t) &= \frac{1}{R_{0}^{[k]}(Nt)} \sum_{\ell=1}^{k+1} \frac{K_\ell}{N^{\ell-1}}\\
&\hspace{-1cm}\times \sum_{\zeta\in B_{k}(0)}
R_{\zeta}^{[\ell]}(N^{1+k-\ell}t) \sin\big[\Phi_{\zeta}^{[\ell]}(N^{k-\ell}t) - \theta_{\zeta}(N^{k}t)\big]
\cos\big[\Phi_{0}^{[k]}(t) - \theta_{\zeta}(N^{k}t)\big],\\
\ddd J_{N}^{[k]}(t) &= \frac{1}{R_{0}^{[k]}(Nt)}\,\sqrt{Q_{0}^{[k]}(Nt)}\,\ddd W^{[k]}(t).
\end{aligned}
\end{equation}
In the last line we use that $(W_\zeta(t))_{t \geq 0}$, $\zeta \in B_k(0)$, are i.i.d.\ and write
$(W^{[k]}(t))_{t \geq 0}$ to denote an auxiliary Brownian motion associated with level $k$.

The \emph{truncation approximation} consists of throwing away the terms with $1 \leq \ell \leq k$ 
and keeping only the terms with $\ell=k+1$.

%%%

\subsubsection{$\bullet$ Level $k=1$}
\label{subsec:k=1}

For $k=1$, by \eqref{eq:embedding} the first line of \eqref{IJkNt} reads
\begin{align}
&I_{N}^{[1]}(t) = K_{1}\,\sum_{\zeta\in B_{1}(0)} \sin\big[\Phi_{0}^{[1]}(t) - \theta_{\zeta}(Nt)\big]
\cos\big[\Phi_{0}^{[1]}(t) - \theta_{\zeta}(Nt)\big]\\ \nonumber
&\qquad + K_{2}\,\frac{R_{0}^{[2]}(t)}{R_{0}^{[1]}(Nt)}\,\frac{1}{N} \sum_{\zeta\in B_{1}(0)} 
\sin\big[\Phi_{0}^{[2]}(N^{-1}t) - \theta_{\zeta}(Nt)\big]
\cos\big[\Phi_{0}^{[1]}(t) - \theta_{\zeta}(Nt)\big].
\end{align}
We telescope the sine. Using \eqref{trig} with $a=\Phi_{0}^{[2]}(N^{-1}t)-\Phi_{0}^{[1]}(t)$ 
and $b=\Phi_{0}^{[1]}(t)-\theta_{\zeta}(Nt)$, we obtain
\begin{align}
\label{telescope1}
I_{N}^{[1]}(t) &= K_{1}\,\sum_{\zeta\in B_{1}(0)} \sin\big[\Phi_{0}^{[1]}(t) - \theta_{\zeta}(Nt)\big]
\cos\big[\Phi_{0}^{[1]}(t) - \theta_{\zeta}(Nt)\big]\\ \nonumber
&\qquad + K_{2}\,\frac{R_{0}^{[2]}(t)}{R_{0}^{[1]}(Nt)}
\sin\big[\Phi_{0}^{[2]}(N^{-1}t) - \Phi_{0}^{[1]}(t)\big]\\ \nonumber 
&\qquad\qquad \times \frac{1}{N}
\sum_{\zeta\in B_{1}(0)}\cos^{2}\big[\Phi_{0}^{[1]}(t) - \theta_{\zeta}(Nt)\big]\\ \nonumber
&\qquad + K_{2}\,\frac{R_{0}^{[2]}(t)}{R_{0}^{[1]}(Nt)}\
\cos\big[\Phi_{0}^{[2]}(N^{-1}t) - \Phi_{0}^{[1]}(t)\big]\\ \nonumber
&\qquad\qquad \times \frac{1}{N} \sum_{\zeta\in B_{1}(0)}\sin\big[\Phi_{0}^{[1]}(t) - \theta_{\zeta}(Nt)\big]
\cos\big[\Phi_{0}^{[1]}(t) - \theta_{\zeta}(Nt)\big].
\end{align}
On time scale $Nt$, the oscillators in the 1-block have synchronized, and hence the last sum
vanishes in the limit $N\to\infty$ by the symmetry property in \eqref{eq:equisym} for $k=1$. 
Therefore we have
\begin{align}
\label{I1Ntrewrite}
I_{N}^{[1]}(t) &= K_{1}\,\sum_{\zeta\in B_{1}(0)} \sin\big[\Phi_{0}^{[1]}(t) - \theta_{\zeta}(Nt)\big]
\cos\big[\Phi_{0}^{[1]}(t) - \theta_{\zeta}(Nt)\big]\\ \nonumber
&\qquad + K_{2}\,\frac{R_{0}^{[2]}(t)\,Q^{[1]}_{0}(Nt)}{R_{0}^{[1]}(Nt)}
\sin\big[\Phi_{0}^{[2]}(N^{-1}t) - \Phi_{0}^{[1]}(t)\big] + o(1).
\end{align}
Recalling \eqref{IJkNt} we further have
\begin{equation}
\label{JNrewrite}
\ddd J_{N}^{[1]}(t) = \frac{1}{R_{0}^{[1]}(Nt)}\,\sqrt{Q_{0}^{[1]}(Nt)}\,\ddd W^{[1]}(t)
\end{equation}
with 
\begin{equation}
Q_{0}^{[1]}(Nt) = \frac{1}{N} \sum_{\zeta \in B_{1}(0)} \cos^2\big[\Phi_{0}^{[1]}(t) - \theta_{\zeta}(Nt)\big].
\end{equation}
The first term in the right-hand side of \eqref{I1Ntrewrite} is the same as that in \eqref{eq:maindiff*} 
with $K=K_1$ and $\psi_N(Nt)= \Phi_{0}^{[1]}(t)$. The term in the right-hand side of \eqref{JNrewrite} 
is the same as that of \eqref{eq:maindiff} with $r_N(Nt)=R_{0}^{[1]}(Nt)$ and $W_*(t) = W^{[1]}(t)$. 
Together they produce, in the limit as $N\to\infty$, the same noise term as in the mean-field model, 
namely,
\begin{equation}
\mathcal{D}^{[1]}\,\ddd W^{[1]}(t) 
\end{equation}
with a \emph{renormalized noise strength} 
\begin{equation}
\label{eq:D*use}
\mathcal{D}^{[1]} = D_*(K_1)
\end{equation}
given by \eqref{D*def} with $D=1$, where we use that
\begin{equation}
\lim_{N\to\infty} R_{0}^{[1]}(Nt) = R^{[1]} = R^{[1]}(K_1),
\quad \lim_{N\to\infty} Q_{0}^{[1]}(Nt) = Q^{[1]} = Q^{[1]}(K_1) \quad \forall\; t>0.
\end{equation}
The second term in the right-hand side of \eqref{I1Ntrewrite} is precisely the Kuramoto-type 
interaction term of $\Phi_{0}^{[1]}(t)$ with the average phase of the oscillators in the 2-block 
at time $Nt$. Therefore, in the limit as $N\to\infty$, we end up with the limiting SDE 
\begin{align}
\ddd\Phi_{0}^{[1]}(t) = K_{2}\,\mathcal{E}^{[1]}\,R_{0}^{[2]}(t)\,
\sin\big[\Phi - \Phi_{0}^{[1]}(t)\big] + \mathcal{D}^{[1]}\,\ddd W^{[1]}(t)
\end{align}
with 
\begin{equation}
\mathcal{E}^{[1]} = \frac{Q^{[1]}}{R^{[1]}}.
\end{equation}
If we leave out the first term in the right-hand side of \eqref{I1Ntrewrite} (which as shown 
in \eqref{ratioD} may be done at the cost of an error of less than 4 percent), then we end 
up with the limiting SDE 
\begin{align}
\ddd\Phi_{0}^{[1]}(t) = K_{2}\,\bar{\mathcal{E}}^{[1]}\,R_{0}^{[2]}(t)\,
\sin\big[\Phi - \Phi_{0}^{[1]}(t)\big] + \bar{\mathcal{D}}^{[1]}\,\ddd W^{[1]}(t)
\end{align}
with $\bar{\mathcal{E}}^{[1]}=\mathcal{E}^{[1]}$ and
\begin{equation}
\label{eq:barD*use}
 \bar{\mathcal{D}}^{[1]} = \bar D_*(K_1) = \frac{\sqrt{Q^{[1]}}}{R^{[1]}}
\end{equation}
given by \eqref{barD*def} with $D=1$. Thus we have justified the SDE in \eqref{eq:thm:scalphasendapprox} 
for $k=1$. After a transient period we have $\lim_{t\to\infty} R_{0}^{[2]}(t) = R_{0}^{[2]}$. 

Note that, in the approximation where we leave out the first term in the right-hand side of \eqref{I1Ntrewrite},
the pair $(R^{[1]},Q^{[1]})$ takes over the role of the pair $(r,q)$ in the mean-field model. The latter are the 
unique solution of the consistency relation and recursion relation (recall \eqref{eq:Zdef}, \eqref{eq:plt},
\eqref{eq:pltalt} and \eqref{eq:qrdef})
\begin{equation}
r = \frac{I_1(2Kr)}{I_0(2Kr)}, \qquad q =  \frac{I_2(2Kr)}{I_0(2Kr)}.
\end{equation}
These can be summarised as saying that $(r,q)=\cT_K(1,1)$, with $\cT_K$ the renormalization map 
introduced in Definition~\ref{def:TKnd}. Thus we see that 
\begin{equation}
\label{eq:TKstage}
(R^{[1]},Q^{[1]})=\cT_{K_1}(1,1),
\end{equation} 
which explains why $\cT_{K_1}$ comes on stage.

%%%

\subsubsection{$\bullet$ Levels $k \geq 2$}
\label{subsec:k=2}

For $k \geq 2$, by \eqref{eq:embedding} the term with $\ell=k+1$ in $I_{N}^{[k]}(t)$ in the first line 
of \eqref{IJkNt} equals
\begin{equation}
\begin{aligned}
&I_{N}^{[k]}(t)|_{\ell=k+1}\\ 
&= K_{k+1}\,\frac{R_{0}^{[k+1]}(t)}{R_{0}^{[k]}(Nt)}\,\frac{1}{N^k} \sum_{\zeta\in B_{k}(0)} 
\sin\big[\Phi_{0}^{[k+1]}(N^{-1}t) - \theta_{\zeta}(N^kt)\big]
\cos\big[\Phi_{0}^{[k]}(t) - \theta_{\zeta}(N^kt)\big].
\end{aligned}
\end{equation}
We again telescope the sine. Using \eqref{trig}, this time with $a=\Phi_{0}^{[k+1]}(N^{-1}t) 
-\Phi_{0}^{[k]}(t)$ and $b=\Phi_{0}^{[k]}(t)-\theta_{\zeta}(N^kt)$, we can write
\begin{equation}
\begin{aligned}
I_{N}^{[k]}(t)|_{\ell=k+1}
&= K_{k+1}\,\frac{R_{0}^{[k+1]}(t)}{R_{0}^{[k]}(Nt)}
\sin\big[\Phi_{0}^{[k+1]}(N^{-1}t)-\Phi_{0}^{[k]}(t)\big]\\
&\qquad\qquad \times \frac{1}{N^k} \sum_{\zeta\in B_{k}(0)} 
\cos^2\big[\Phi_{0}^{[k]}(t)-\theta_{\zeta}(N^kt)\big]\\
&\qquad+ K_{k+1}\,\frac{R_{0}^{[k+1]}(t)}{R_{0}^{[k]}(Nt)}
\sin\big[\Phi_{0}^{[k+1]}(N^{-1}t)-\Phi_{0}^{[k]}(t)\big]\\
&\qquad\qquad \times \frac{1}{N^k} \sum_{\zeta\in B_{k}(0)} 
\sin\big[\Phi_{0}^{[k]}(t)-\theta_{\zeta}(N^kt)\big]\cos\big[\Phi_{0}^{[k]}(t)-\theta_{\zeta}(N^kt)\big].
\end{aligned}
\end{equation}
By the symmetry property in \eqref{eq:equisym}, the last term vanishes as $N\to\infty$, 
and so we have
\begin{equation}
I_{N}^{[k]}(t)|_{\ell=k+1} = K_{k+1}\,\frac{R_{0}^{[k+1]}(t)\,Q^{[k]}_{0}(Nt)}{R_{0}^{[k]}(Nt)}\,
\sin\big[\Phi_{0}^{[k+1]}(N^{-1}t) - \Phi_{0}^{[k]}(t)\big] + o(1).
\end{equation}
Using that 
\begin{equation}
\lim_{N\to\infty} R_{0}^{[k]}(Nt) = R^{[k]},
\qquad \lim_{N\to\infty} Q_{0}^{[k]}(Nt) = Q^{[k]} \qquad \forall\,t>0,
\end{equation}
we obtain
\begin{equation}
I_{N}^{[k]}(t)|_{\ell=k+1} = K_{k+1}\,\frac{Q^{[k]}}{R^{[k]}}\,R_{0}^{[k+1]}(t)\,
\sin\big[\Phi - \Phi_{0}^{[k]}(t)\big] + o(1),
\end{equation}
which is the Kuramoto-type interaction term of $\Phi_{0}^{[k]}(t)$ with the average phase of 
the oscillators in the $(k+1)$-block at time $N^kt$. The noise term in \eqref{IJkNt} scales like
\begin{equation}
\ddd J_{N}^{[k]}(t) = \frac{1}{R^{[k]}}\,\sqrt{Q^{[k]}}\,\ddd W^{[k]}(t) + o(1).
\end{equation}
Hence we end up with
\begin{equation}
I_{N}^{[k]}(t)|_{\ell=k+1} \ddd t + \ddd J_{N}^{[k]}(t) 
= K_{k+1}\,\frac{Q^{[k]}}{R^{[k]}}\,R_{0}^{[k+1]}(t)\,
\sin\big[\Phi - \Phi_{0}^{[k]}(t)\big] + \frac{\sqrt{Q^{[k]}}}{R^{[k]}}\,\ddd W^{[k]}(t) + o(1).
\end{equation}
Thus we have justified the SDE in \eqref{eq:thm:scalphasendapprox} for $k \geq 2$, with
$\bar{\mathcal{E}}^{[k]}$ and $\bar{\mathcal{D}}^{[k]}$ given by \eqref{eq:EDdef}. Note
that 
\begin{equation}
(R^{[k]},Q^{[k]}) = \cT_{K_k}(R^{[k-1]},Q^{[k-1]}), 
\end{equation}
in full analogy with \eqref{eq:TKstage}.

For $k \geq 2$ the term with $\ell=k$ equals
\begin{equation}
\label{I1k}
I_{N}^{[k]}(t)|_{\ell=k} = K_k \sum_{i=1}^N \frac{1}{N^{k-1}} \sum_{\zeta \in B_{k-1}(i)}\\
\sin\big[\Phi_{0}^{[k]}(t) - \theta_{\zeta}(N^{k}t)\big]
\cos\big[\Phi_{0}^{[k]}(t) - \theta_{\zeta}(N^{k}t)\big],
\end{equation}
where $B_{k-1}(i)$, $1 \leq i \leq N$, are the $(k-1)$-blocks making up the $k$-block $B_k(0)$,
and we use that $(R_{\zeta}^{[k]}(t),\Phi_{\zeta}^{[k]}(t)) = (R_{0}^{[k]}(t),\Phi_{0}^{[k]}(t))$ for
all $\zeta \in B_{k-1}(i)$ and all $1 \leq i \leq N$. The sum in \eqref{I1k} has a similar form as 
the first term in the right-hand side of \eqref{I1Ntrewrite}, but now with the $1$-block replaced 
by $N$ copies of $(k-1)$-blocks. This opens up the possibility of a finer approximation analogous 
to the one obtained by using \eqref{eq:D*use} instead of \eqref{eq:barD*use}. As we argued in 
Section~\ref{sec:prep}, the improvement should be minor (recall \eqref{ratioD}).

%%%%%%%%%%%%%%%%%% SECTION 4 %%%%%%%%%%%%%%%%%%%%%%%%

\section{Universality classes and synchronization levels}
\label{sec:thmclasses}

In Section~\ref{sec:renmapprop} we derive some basic properties of the renormalization map 
(Lemmas~\ref{lem:kincreasingR}--\ref{lem:components} below). In Section~\ref{sec:nodisorder} 
we prove Theorem~\ref{thm:classesnd}. The proof relies on convexity and sandwich estimates 
(Lemmas~\ref{lem:Zprop}--\ref{lem:sand} below).

%%%%%%%%%%%%%%%%%%%%%%%%%%%%%%%%%%%

\subsection{Properties of the renormalization map}
\label{sec:renmapprop}

For $\lambda \in [0,\infty)$, define
\begin{equation}
\label{eq:VWdef}
\begin{aligned}
V(\lambda) &= \int_0^{2\pi} \ddd\theta\,\cos\theta\,p_{\lambda}(\theta) = \frac{I_{1}(\lambda)}{I_{0}(\lambda)},\\
W(\lambda) &= \int_0^{2\pi} \ddd\theta\,\cos^{2}\theta\,p_{\lambda}(\theta)= \frac{I_{2}(\lambda)}{I_{0}(\lambda)},
\end{aligned}
\end{equation}
where the probability distribution $p_\lambda(\theta)$ is given by \eqref{eq:stat} with 
$\omega \equiv 0$ and $D=1$. The renormalization map $\cT_K$ in \eqref{eq:TKmap} 
can be written as $(\bar{R},\bar{Q})=\cT_K(R, Q)$ with
\begin{align}
\label{eq:RVrel}
\bar{R} &= RV(\lambda),\nonumber\\
\bar{Q}-\tfrac12 &= (Q-\tfrac12)\big[2W(\lambda) - 1\big],
\end{align}
and $\lambda = 2K\bar{R}\sqrt{Q}$. It is known that $\lambda \mapsto V(\lambda)$ is
strictly increasing and strictly convex, with $V(0)=0$ and $\lim_{\lambda\to\infty} V(\lambda) 
= 1$. 

\begin{lemma}
\label{lem:kincreasingR}
The map $K \mapsto \bar{R}(R, K)$ is strictly increasing.
\end{lemma}

\begin{proof}
The derivative of $\bar{R}$ w.r.t.\ $K$ exists by the implicit function theorem, so that
\begin{eqnarray}
\label{eq:kincrease}
&&\frac{d\bar{R}}{d K} = 2RV'(2K\bar{R})\,\left[\bar{R} + K \frac{d\bar{R}}{d K}\right],\nonumber\\
&&\frac{d\bar{R}}{d K}\,\big[1-2KRV'(2K\bar{R})\big] = 2R\bar{R}V'(2K\bar{R}). 
\end{eqnarray}
Note that $\bar{R}$ is the solution to $\bar{R} = RV(2K\bar{R})$, which is non-trivial only when 
$1<2RKV'(2K\bar{R})$ due to the concavity of the map $R \mapsto RV(2K\bar{R})$. This implies 
that $2KRV'(2K\bar{R})<1$ at the solution, which makes the term in \eqref{eq:kincrease} 
between square brackets positive. The claim follows since we proved previously that 
$R, \bar{R} \in [0,1)$ and $V'(2K\bar{R})>0$. 
\end{proof}

\begin{lemma}\label{lem:kincreasingQ}
The map $K \mapsto \bar{Q}(\bar{R}, K, Q)$ is strictly increasing.
\end{lemma}

\begin{proof}
The derivative of $\bar{Q}$ w.r.t.\ $K$ exists by the implicit function theorem, so that
\begin{equation}
\frac{d\bar{Q}}{dK} = (Q - \tfrac{1}{2})\,4\sqrt{Q}\,
W'\big(2\sqrt{Q}K\bar{R}\big)\,\left[\bar{R} + K\frac{d\bar{R}}{dK}\right].
\end{equation}
We have that $(Q - \tfrac{1}{2})\sqrt{Q} \geq 0$ because $Q\in [\frac{1}{2}, 1)$, $W'(2\sqrt{Q}K
\bar{R})>0$ as proven before, and $[\bar{R} + K\frac{d\bar{R}}{dK}]>0$ as in the proof of 
Lemma~\ref{lem:kincreasingR}. The claim therefore follows. 
\end{proof}

\begin{lemma}
\label{lem:components}
The map $(R,Q)\mapsto (\bar{R}, \bar{Q})$ is non-increasing in both components, 
i.e., 
\begin{enumerate}[(i)]
\item[{\rm (i)}] $R \mapsto \bar{R}(K, R)$ is non-increasing.
\item[{\rm (ii)}] $Q \mapsto \bar{Q}(K, \bar{R}, Q)$ is non-increasing.
\end{enumerate}
\end{lemma}

\begin{proof}
(i) We have 
\begin{equation}
\bar{R} = R\,V\big(2\sqrt{Q}K\bar{R}\big).
\end{equation}
But $V(\sqrt{Q}K\bar{R})\in [0,1)$, and so $\bar{R}\leq R$.\\
(ii) We have
\begin{equation}
\bar{Q} - \tfrac{1}{2} = (Q-\tfrac{1}{2})\,\big[2W\big(2\sqrt{Q}K\bar{R}\big)-1\big].
\end{equation}
But $W(2\sqrt{Q}K\bar{R})\in [\frac{1}{2},1)$, and so $\bar{Q}\leq Q$. In fact, since both 
$V(2\sqrt{Q}K\bar{R})$ and $W(2\sqrt{Q}K\bar{R})$ are $<1$, both maps are strictly decreasing 
until $R=0$ and $Q = \tfrac12$ are hit, respectively.
\end{proof}

%%%%%%%%%%%%%%%%%%%%%%%%%%%%%%%%%%%

\subsection{Renormalization}
\label{sec:nodisorder}

Recall \eqref{eq:Zdef}. To prove Theorems~\ref{thm:classesnd} we need the following lemma.

\begin{lemma}
\label{lem:Zprop}
The map $\lambda \mapsto \log I_{0}(\lambda)$ is analytic, strictly increasing and strictly convex 
on $(0,\infty)$, with
\begin{equation}
\label{eq:Zlims}
I_{0}(\lambda) = 1 + \tfrac14\lambda^2\,[1 + O(\lambda^2)], 
\quad \lambda \downarrow 0,
\qquad I_{0}(\lambda) = \frac{\eee^{\lambda}}{\sqrt{2\pi\lambda}}\,[1 + O(\lambda^{-1})], 
\quad \lambda\to\infty.
\end{equation}
\qed
\end{lemma}

\begin{proof}
Analyticity is immediate from \eqref{eq:Zdef}. Strict convexity follows because the 
numerator of $[\log I_{0}(\lambda)]''$ equals 
\begin{equation}
\begin{aligned}
I_{2}(\lambda)I_{0}(\lambda)-I_{1}(\lambda)I_{1}(\lambda) 
&= \frac{1}{2\pi}\int_{0}^{2\pi}\ddd\phi\int_{0}^{2\pi}\ddd\psi\,
[\cos^{2}\phi - \cos\phi\cos\psi]\,\eee^{\lambda(\cos\phi+\cos\psi)}\\
&= \frac{1}{2\pi}\int_{0}^{2\pi}\ddd\phi\int_{0}^{2\pi}\ddd\psi\,
[\cos\phi - \cos\psi]^2\,\eee^{\lambda(\cos\phi+\cos\psi)} > 0,
\end{aligned}
\end{equation}
where we symmetrize the integrand. Since $\log I_{0}(0)=0$, $\log I_{0}(\lambda)>0$ for 
$\lambda>0$ and $\lim_{\lambda\to\infty} \log I_{0}(\lambda) = \infty$, the strict monotonicity 
follows. The asymptotics in \eqref{eq:Zlims} is easily deduced from \eqref{eq:Zdef} via a 
saddle point computation. 
\end{proof}

\noindent
Since $V=I_{1}/I_{0}=[\log I_{0}]'$, we obtain from \eqref{eq:Zlims} and convexity that
\begin{eqnarray}
\label{VWscal1}
V(\lambda) &\sim& \tfrac12\lambda, \qquad \lambda \downarrow 0,\\
\label{VWscal2}
1-V(\lambda) &\sim& \frac{1}{2 \lambda}, \qquad \lambda \to \infty.
\end{eqnarray}
This limiting behaviour of $V(\lambda)$ inspires the choice of bounding functions in the 
next lemma.

\begin{lemma}
\label{lem:Vplusminus}
$V^{+}(\lambda) \geq V(\lambda) \geq V^{-}(\lambda)$ for all $\lambda \in (0,\infty)$ with
(see Fig.~{\rm \ref{fig:bounding}})
\begin{equation}
\label{eq:Vbdsdef}
\begin{aligned}
V^{+}(\lambda) &= \frac{2\lambda}{1+2\lambda},\\ 
V^{-}(\lambda) &= \frac{\tfrac12\lambda}{1+\tfrac12\lambda}.
\end{aligned}
\end{equation}
\qed
\end{lemma}

\begin{proof}
Segura~\cite[Theorem 1]{S11} shows that
\begin{equation}
V(\lambda) < V_*^{+}(\lambda) 
= \frac{\lambda}{\tfrac12 + \sqrt{(\tfrac12)^{2} + \lambda^{2}}}, \qquad \lambda>0.
\end{equation} 
Since $\lambda < \sqrt{(\tfrac12)^2 + \lambda^{2}}$, it follows that $V_*^{+}(\lambda)<V^{+}
(\lambda)$. Laforgia and Natalini~\cite[Theorem 1.1]{L10} show that
\begin{equation}
V(\lambda) > V_*^{-}(\lambda) = \frac{-1 + \sqrt{\lambda^{2} + 1}}{\lambda}, \qquad \lambda>0.
\end{equation} 
Abbreviate $\eta=\sqrt{\lambda^{2}+1}$. Then $\lambda = \sqrt{(\eta-1)(\eta+1)}$, and we can 
write
\begin{equation}
V_*^{-}(\lambda) = \sqrt{\frac{\eta-1}{\eta+1}} =  \frac{\lambda}{\eta+1} = \frac{\lambda}{2+(\eta-1)}.
\end{equation}
Since $\lambda>\eta-1$, it follows that $V_*^{-}(\lambda)>V^{-}(\lambda)$.
\end{proof}

%%%%%%%%%%%%%%%%%%%%%%%%%%%%%%%%%%%%%%%
\begin{figure}[htbp]
\centering
\includegraphics[scale=0.8]{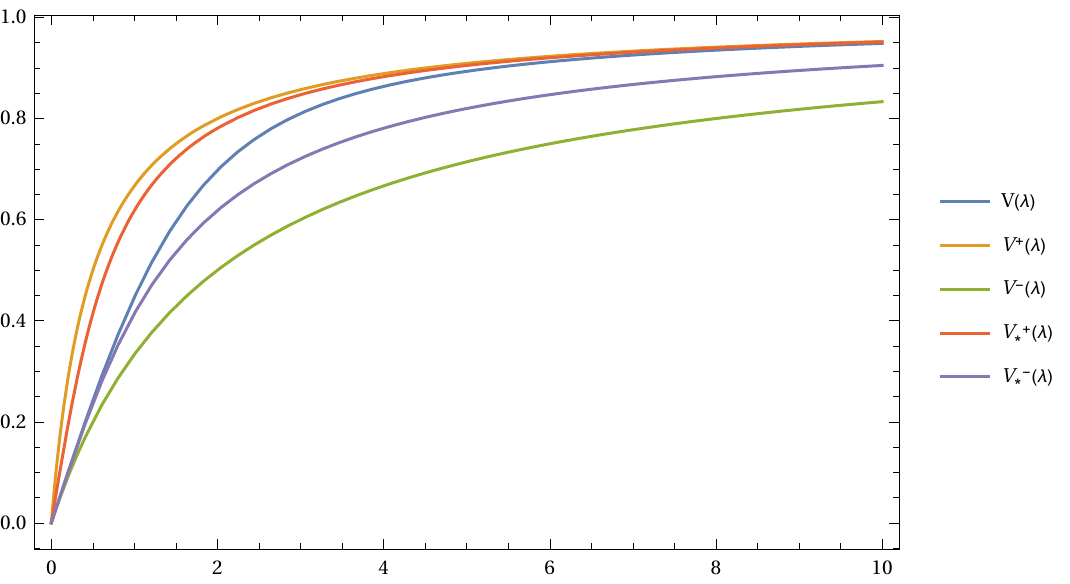}
\caption{Plots of the tighter bounds in the proof of Lemma~\ref{lem:Vplusminus} 
and the looser bounds needed for the proof of Theorem~\ref{thm:classesnd}.}
\label{fig:bounding}
\end{figure}
%%%%%%%%%%%%%%%%%%%%%%%%%%%%%%%%%%%%%%%

\noindent
Note that both $V^{+}$ and $V^{-}$ are strictly increasing and concave on $(0,\infty)$, 
which guarantees the uniqueness and non-triviality of the solution to the consistency 
relation in the first line of \eqref{eq:RVrel} when we replace $V(\lambda)$ by either 
$V^{+}(\lambda)$ or $V^{-}(\lambda)$. 

In the sequel we write $V,W,R_k,Q_k$ instead of $V_{\delta_0},W_{\delta_0},R^{[k]},Q^{[k]}$ 
to lighten the notation. We know that $(R_k)_{k\in\N_0}$ is the solution of the sequence of 
consistency 
relations 
\begin{equation}
\label{eq:conrelseq}
R_{k+1} = R_k V\big(2\sqrt{Q_k} K_{k+1}R_{k+1}\big), \qquad k \in \N_0.
\end{equation}
This requires as input the sequence $(Q_k)_{k\in\N_0}$, which is obtained from the sequence 
of recursion relations 
\begin{equation}
Q_{k+1}-\tfrac12 = (Q_k-\tfrac12)\big[2W\big(2\sqrt{Q_k} K_{k+1}R_{k+1}\big)-1\big].  
\end{equation}
By using that $Q_k \in [\frac12,1]$ for all $k\in\N_0$, we can remove $Q_k$ from \eqref{eq:conrelseq}
at the cost of doing estimates. Namely, let $(R_{k}^{+})_{k\in\N_0}$ and $(R_{k}^{-})_{k\in\N_0}$ denote 
the solutions of the sequence of consistency relations 
\begin{equation}
\label{eq:conrelseqbds}
\begin{array}{lll}
R^+_{k+1} &= R_k V^+\big(2K_{k+1}R^+_{k+1}\big), &k \in \N_0,\\
R^-_{k+1} &= R_k V^-\big(2\sqrt{\tfrac12}K_{k+1}R^-_{k+1}\big), &k \in \N_0.
\end{array}
\end{equation}

\begin{lemma}
\label{lem:sand}
$R_k^{+} \geq R_k \geq R^{-}_k$ for all $k\in\N$. \qed
\end{lemma}

\begin{proof}
If we replace $V(\lambda)$ by $V^{+}(\lambda)$ (or $V^{-}(\lambda)$) in the consistency 
relation for $R_{k+1}$ given by \eqref{eq:conrelseq}, then the new solution $R_{k+1}^{+}$ 
(or $R^{-}_{k+1}$) is larger (or smaller) than $R_{k+1}$. Indeed, we have
\begin{equation}
\label{eq:Rbd}
R_{k+1} = R_{k}V(2K_{k+1}R_{k+1}\sqrt{Q_{k}}) \leq R_{k}V^{+}(2K_{k+1}R_{k+1}).
\end{equation}
Because $V^+$ is concave, it follows from \eqref{eq:Rbd} and the first line of \eqref{eq:conrelseqbds}
that $R_{k+1} \leq R^+_{k+1}$.
\end{proof}

We are now ready to prove Theorems~\ref{thm:classesnd}--\ref{thm:critcasend}.

\begin{proof}
From the first lines of \eqref{eq:Vbdsdef} and \eqref{eq:conrelseqbds} we deduce
\begin{equation}
R_k > \frac{1}{4K_{k+1}} \quad \Longleftrightarrow \quad R_{k+1}^+>0 
\quad \Longrightarrow \quad 
R_{k} - R_{k+1}^{+} = \frac{1}{4K_{k+1}}.
\end{equation}
Hence, with the help of Lemma~\ref{lem:sand}, we get
\begin{equation}
R_k > \frac{1}{4K_{k+1}} \quad \Longrightarrow \quad 
R_{k} - R_{k+1} \geq \frac{1}{4K_{k+1}}.
\end{equation}
Iteration gives (recall that $R_0=1$) 
\begin{equation}
1- R_{k} \geq \min\left\{1, \sum_{\ell=1}^{k}\frac{1}{4K_{\ell}}\right\}.
\end{equation}
As soon as the sum in the right-hand side is $\geq 1$, we know that $R_{k}=0$. This 
gives us the criterion for universality class (1) in Theorem~\ref{thm:classesnd}. Similarly, 
from the second lines of \eqref{eq:Vbdsdef} and \eqref{eq:conrelseqbds} we deduce 
\begin{equation}
\label{eq:Rbdalt}
R_k > \frac{2\sqrt{2}}{K_{k+1}} \quad \Longleftrightarrow \quad R_{k+1}^->0 
\quad \Longrightarrow \quad R_{k} - R_{k+1}^{-} = \frac{\sqrt{2}}{K_{k+1}}.
\end{equation}
Hence, with the help of Lemma~\ref{lem:sand}, we get 
\begin{equation}
R_k > \frac{\sqrt{2}}{K_{k+1}} \quad \Longrightarrow \quad
R_{k} - R_{k+1} \leq  \frac{\sqrt{2}}{K_{k+1}}.
\end{equation}
Iteration gives
\begin{equation}
1- R_{k} \leq \max\left\{1,\sum_{\ell=1}^{k} \frac{\sqrt{2}}{K_{\ell}}\right\}.
\end{equation}
As soon as the sum in the right-hand side is $<1$, we know that $R_{k}>0$. This gives us 
the criterion for universality class (3) in Theorem~\ref{thm:classesnd}. 

In universality classes (2) and (3) we have $R_{k}^+ \geq R_{k}>0$ for $k\in\N$, and 
hence
\begin{equation}
\label{eq:R1}
R_{k}-R_{\infty} = \sum_{\ell \geq k} (R_{\ell}-R_{\ell+1}) 
\geq \sum_{\ell \geq k} (R_{\ell}-R_{\ell+1}^+)
= \sum_{\ell \geq k}\frac{1}{4K_{\ell+1}}, \qquad k \in \N_0.
\end{equation}
In universality class (1), on the other hand, we have $R_{k}^+ \geq R_{k}>0$ for $1 \leq k < k_*$ 
and $R_k=0$ for $k \geq k_*$, and hence
\begin{equation}
\label{eq:R2}
R_{k} - R_{k_*-1} = \sum_{\ell=k}^{k_*-2} (R_{\ell}-R_{\ell+1})
\geq \sum_{\ell=k}^{k_*-2} (R_{\ell}-R_{\ell+1}^+)
= \sum_{\ell=k}^{k_*-2} \frac{1}{4K_{\ell+1}}, \qquad 0 \leq k \leq k_*-2.
\end{equation}
Finally, with no assumption on $(R_k)_{k\in\N}$, we have
\begin{equation}
\label{eq:R3}
R_{k} - R_{\infty} = \sum_{\ell \geq k} (R_{\ell}-R_{\ell+1}) 
\leq \sum_{\ell \geq k} (R_{\ell}-R_{\ell+1}^-) 
\leq \sum_{\ell \geq k}\frac{\sqrt{2}}{K_{\ell+1}},
\end{equation}
where the last inequality follows from \eqref{eq:Rbdalt}. The bounds in \eqref{eq:R1}--\eqref{eq:R3} 
yields the sandwich in Theorem~\ref{thm:critcasend}.
\end{proof}

\begin{remark}
{\rm In the proof of Theorem~\ref{thm:classesnd}--\ref{thm:critcasend} we exploited the fact that 
$Q_k \in [\tfrac12,1]$ to get estimates that involve a consistency relation in only $R_k$. In principle 
we can improve these estimates by exploring what effect $Q_{k}$ has on $R_{k}$. Namely, in 
analogy with Lemma~\ref{lem:Vplusminus}, we have $W^{+}(\lambda) \geq W(\lambda) \geq 
W^{-}(\lambda)$ for all $\lambda \in (0,\infty)$ with (see Fig.~\ref{fig:sandw})
\begin{equation}
W^{+}(\lambda) = \frac{1+\lambda}{2+\lambda}, 
\qquad W^{-}(\lambda) = \frac{1 - \lambda + \lambda^{2}}{2 + \lambda^{2}}.
\end{equation}
This allows for better control on $Q_k$ and hence better control on $R_k$. However, the formulas
are cumbersome to work with and do not lead to a sharp condition anyway.} \qed
\end{remark} 

%%%%%%%%%%%%%%%%%%%%%%%%%%%%%%%%%%%%%%%%%
\begin{figure}[H]
\vspace{0.2cm}
\centering
\includegraphics[scale=0.8]{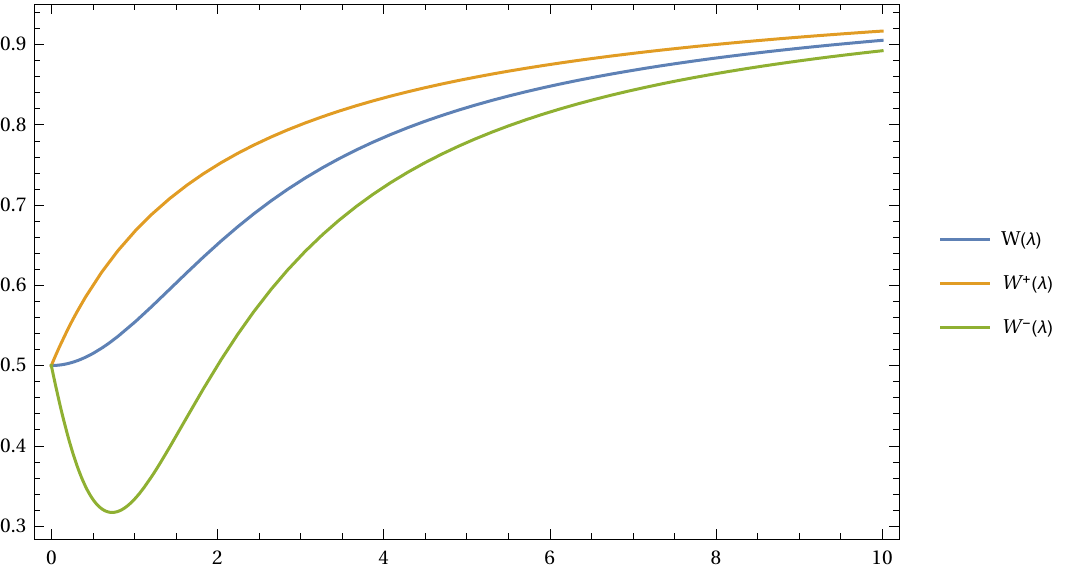}
\caption{Bounding functions for $W(\lambda)$.}
\label{fig:sandw}
\end{figure}
%%%%%%%%%%%%%%%%%%%%%%%%%%%%%%%%%%%%%%

%%%%%%%%%%%%%%%%%%%%%%%%%%%%%%%%%%%%%

\appendix

%%%%%%%%%%%%%%% APPENDIX A %%%%%%%%%%%%%%%%%%%%

\section{Numerical analysis}
\label{app:numerics}

In this appendix we numerically compute the iterates of the renormalization map in \eqref{eq:TKmap}
for two specific choices of $(K_k)_{k\in\N}$, belonging to universality classes (1) and (3), respectively. 

In Fig.~\ref{fig:uniclass1} we show an example in universality class (1): $K_{k} = \frac{3}{2\log 2}
\log(k+1)$. Synchronization is lost at level $k = 4$. When we calculate the sum that appears in 
our sufficient criterion for universality class (1), stated in Theorem~\ref{thm:classesnd}, up to 
level $k=4$, we find that 
\begin{align}
\sum_{k=1}^{4}\frac{2\log 2}{3 \log(k+1)} = 1.70774.
\end{align}
This does not exceed $4$, which shows that our sufficient criterion is not tight. It only gives 
us an upper bound for the level above which synchronization is lost for sure (recall 
\eqref{eq:univ1lb2}), although it may be lost earlier.

%%%%%%%%%%%%%%%%%%%%%%%%%%%%%%%%%%%%%%%%%%%%
\begin{figure}[H]
\centering
\includegraphics[scale=0.8]{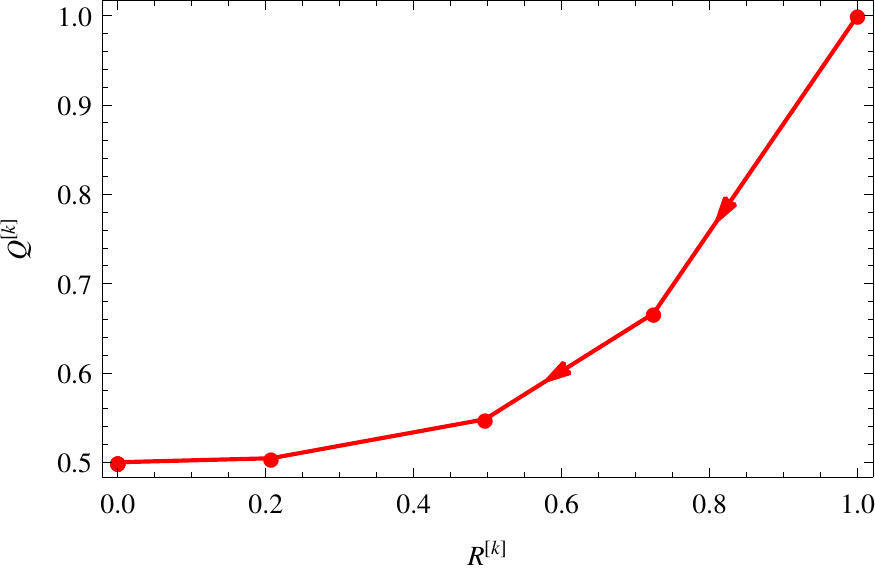}
\includegraphics[scale=0.8]{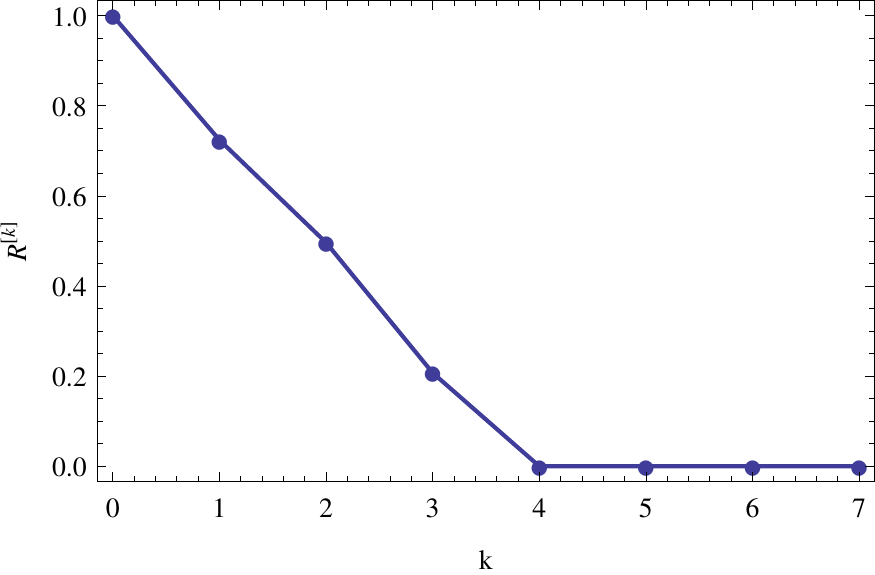}
\caption{A plot of the renormalization map $(R^{[k]}, Q^{[k]})$ for $k=0,\ldots,7$ (left) and 
the corresponding values of $R^{[k]}$ (right) for the choice $K_{k} = \frac{3}{2\log 2}\log(k+1)$.}
\label{fig:uniclass1}
\end{figure}
%%%%%%%%%%%%%%%%%%%%%%%%%%%%%%%%%%%%%%%%%%%%%

In Fig.~\ref{fig:uniclass3} we show an example of universality class (3), where $K_k = 4\,
\eee^{k}$. There is synchronization at all levels. To check our sufficient criterion 
we calculate the sum
\begin{align}
\sum_{k\in\N} \frac{1}{4\,\eee^k} \approx 0.145494 < \frac{1}{\sqrt{2}} \approx 0.7071.
\end{align}

%%%%%%%%%%%%%%%%%%%%%%%%%%%%%%%%%%%%%%%%%%%%%
\begin{figure}[H]
\centering
\includegraphics[scale=0.8]{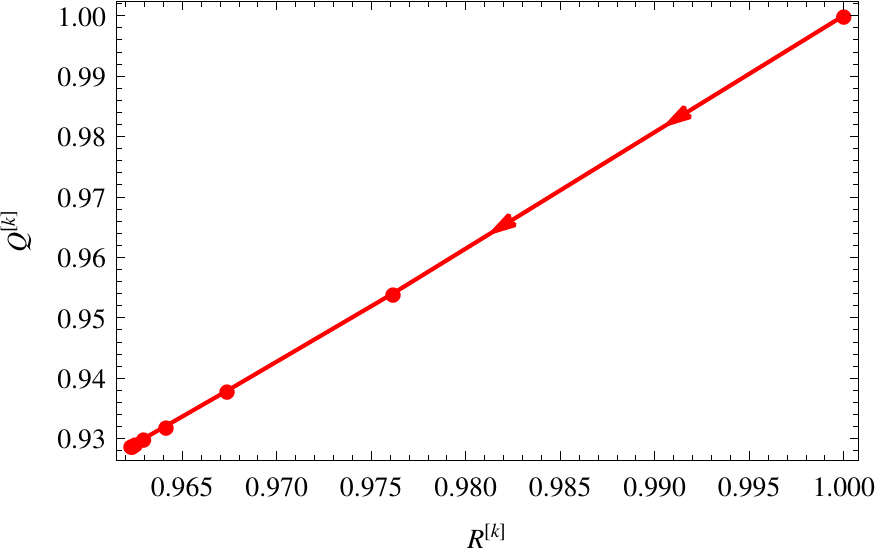}
\includegraphics[scale=0.8]{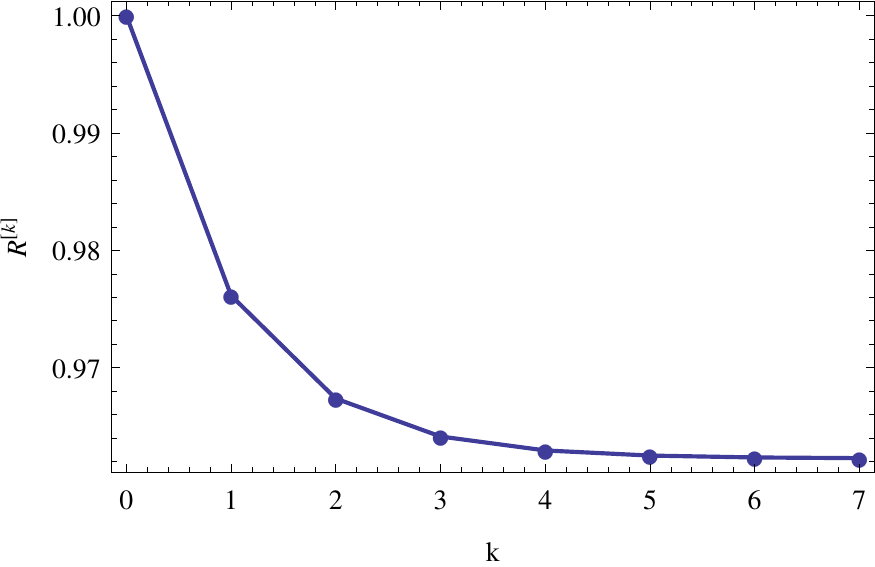}
\caption{A plot of the renormalization map $(R^{[k]}, Q^{[k]})$ for $k=0,\ldots,7$ (left) and the 
corresponding values of $R^{[k]}$ (right) for the choice $K_{k} = 4\,\eee^k$.}
\label{fig:uniclass3}
\end{figure}

To find a sequence $(K_{k})_{k\in\N}$ for universality class (2) is difficult because we do not 
know the precise criterion for criticality. An artificial way of producing such a sequence is to 
calculate the critical interaction strength at each hierarchical level and taking the next interaction 
strength to be $1$ larger.

%%%%%%%%%%% REFERENCES %%%%%%%%%%%%%%%%%%%%%

%%%%%%%%%%%%%%%%%%%%%%%%%%%%%%%%%%%%%%%%%

\end{document}